\documentclass[11  pt, english]{article}   	
\usepackage[margin=1.25in,lmargin=0.75in,rmargin=1.75in,bmargin=0.75in,tmargin=0.75in]{geometry}      
\usepackage{graphicx}
\usepackage{booktabs} 
\usepackage{amsmath}
\usepackage{enumitem}
\usepackage{url}
\usepackage{microtype}
\usepackage{subcaption}
\usepackage{xcolor}
\usepackage{amsthm}
\usepackage{algorithm}
\usepackage[noend]{algpseudocode}
\usepackage{algorithmicx}
\usepackage{commath}
\usepackage{comment}

\usepackage[T1]{fontenc}
\usepackage[utf8]{inputenc}
\usepackage{microtype}
\usepackage{amsmath}
\usepackage{amsfonts}
\usepackage{amssymb}
\usepackage{graphicx}
\usepackage{xspace}
\usepackage[mathscr]{euscript}

\newtheorem{theorem}{Theorem}

\newtheorem{lemma}[theorem]{Lemma}
\newtheorem{definition}{Definition}

\newtheoremstyle{introstyle}
{}
{}
{\itshape}
{}
{\bfseries}
{.}
{ }
{}
\theoremstyle{mystyle}
\newtheorem*{theoremintro}{Theorem}

\usepackage{hyperref}
\hypersetup{
	colorlinks=true, 
	linktoc=all,     
	linkcolor=blue,  
}

\usepackage{mathtools}

\DeclarePairedDelimiter\ceil{\lceil}{\rceil}
\DeclarePairedDelimiter\floor{\lfloor}{\rfloor}

\newcommand{\mchi}{\mathcal{\hat{C}}_i}
\newcommand{\mci}{\mathcal{{C}}_i}

\newcommand{\lcc}{\textbf{LamPrime}}
\newcommand{\lp}{\textbf{LP}}
\newcommand{\opt}{\textbf{OPT}}

\newcommand{\cut}{\textbf{cut}}

\DeclareMathOperator*{\minimize}{minimize}

\DeclareMathOperator{\argmax}{argmax}

\providecommand{\subjectto}{\ensuremath{\text{subject to}}}

\newcommand{\NEPPC}{\ensuremath{\mathit{NEPPC}}}

\newcommand{\eps}{\ensuremath{\varepsilon}}

\newcommand{\tony}[1]{}
\renewcommand{\tony}[1]{{\textcolor{blue}{[{#1}\ --\ AIW]}}}

\newcommand{\lmd}{\lambda}

\newcommand{\curC}{\mathscr{C}}

\newcommand{\mbx}{\mathbf{x}}
\newcommand{\mby}{\mathbf{y}}
\newcommand{\mbc}{\mathbf{c}}
\newcommand{\mbd}{\mathbf{d}}
\newcommand{\mbb}{\mathbf{b}}
\newcommand{\lpcc}{LambdaPrime\xspace}
\newcommand{\orlp}{\textbf{ORLP}$(\mbx^*, s,\lmd, \eps)$\xspace}
\newcommand{\lpPlus}{\textbf{ORLP}$(\mbx^*, 1,\lmd_0, \eps)$\xspace}
\newcommand{\lpMinus}{\textbf{ORLP}$(\mbx^*, -1,\lmd_0, \eps)$\xspace}

\newcommand{\lcclp}{\textbf{LP}$(\lmd)$\xspace}
\newcommand{\lcclpZero}{\textbf{LP}$(\lmd_0)$\xspace}
\newcommand{\lcclpNone}{\textbf{LP}}

\graphicspath{{}} 
\usepackage{multicol}

\usepackage{babel}

\begin{document}
\title{Graph Clustering in All Parameter Regimes}
%
%
%


\author{Junhao Gan \\ The University of Melbourne \\ junhao.gan@unimelb.edu.au \and David F. Gleich\thanks{Supported by NSF grants IIS-1546488, CCF-1909528, CCF-0939370, NSF Center for Science of Information ST, the Sloan Foundation, Purdue IDSI Grant, and NASA.} \\ Purdue University \\dgleich@purdue.edu \and Nate Veldt\thanks{Part of the work was done when the author was visiting the University of Melbourne. Research was supported by the Melbourne School of Engineering.} \\ Cornell University  \\ nveldt@cornell.edu
\and Anthony Wirth \thanks{Supported by the Melbourne School of Engineering} \\
The University of Melbourne \\
awirth@unimelb.edu.au 
\and
Xin Zhang \\
The University of Melbourne \\
xinz11@student.unimelb.edu.au
}

%
%
%
\date{}
\maketitle

\begin{abstract}
Resolution parameters in graph clustering represent a size and quality trade-off. We address the task of efficiently
solving a parameterized graph clustering objective for \emph{all} values of a resolution parameter. Specifically, we
consider an objective we call LambdaPrime, involving a parameter~$\lambda \in (0,1)$. This objective is related to other parameterized clustering problems, such as parametric generalizations of modularity, and captures a number of specific clustering problems as special cases, including sparsest cut and cluster deletion.
While previous work provides approximation results for a single resolution parameter, we seek a set of approximately optimal clusterings for all values of~$\lambda$ in polynomial time. In particular, we ask the question, how small a family of
clusterings suffices to optimize -- or to approximately optimize -- the LambdaPrime objective over the full possible
spectrum of~$\lambda$?

We obtain a family of logarithmically many clusterings by solving the parametric linear
programming relaxation of LambdaPrime at a logarithmic number of parameter values, and round their solutions using
existing approximation algorithms. We prove that this number is tight up to a constant factor. Specifically, for a
certain class of ring graphs, a logarithmic number of feasible solutions is required to provide a constant-factor
approximation for the LambdaPrime LP relaxation in all parameter regimes. We additionally show that for any graph with~$n$
nodes and~$m$ edges, there exists a set of $m$ or fewer clusterings such that for every $\lambda \in (0,1)$, the
family contains an exact solution to the LambdaPrime objective.
There also exists a set of~$O(\log n)$ clusterings that
provide a $(1+\varepsilon)$-approximate solution in all parameter regimes; we demonstrate simple graph classes for which
these bounds are tight.
\end{abstract}

\section{Introduction}
Graph clustering is the task of separating a graph into large, disjoint sets of nodes that share more edges with each
other than the rest of the graph. This often involves, implicitly or explicitly, a trade-off between size and
edge-density. Hence, there are a number of combinatorial objective functions for graph clustering that rely on a tunable
resolution parameter, which  controls the edge density of clusters formed by optimizing the objective. Solving such an
objective for a range of parameters enables detecting different types of clustering structure in the same graph. This has applications to hierarchical clustering~\cite{Delvenne2010stabilitypnas,Jeub2018multiresolution,ReichardtBornholdt2004} and the detection of robust and stable clusterings that remain optimal over a range of parameter settings~\cite{Arenas2008analysis,Delvenne2010stabilitypnas,Jeub2018multiresolution,schaub2012multiscale}. In this paper we develop techniques for finding small families of clusterings that, for all values of the resolution parameter, contain an exact or approximate solution to a graph clustering objective. We also prove fundamental lower bounds on the number of clusterings needed to exactly or approximately optimize a certain objective function for all values of the parameter.  

%

\subsection{Parametric Graph Clustering}
Our work specifically considers a simple parametric clustering objective we call LambdaPrime. This objective can be viewed as a slight variation of the LambdaCC graph clustering framework~\cite{Veldt:2018:CCF:3178876.3186110}, which is itself a parameterized variant of Correlation Clustering~\cite{Bansal2004correlation}, with resolution parameter $\lambda$. Formally, given an undirected graph $G = (V,E)$, the LambdaPrime graph clustering framework partitions $G$ by optimizing the following objective function
\begin{equation}
\label{eq:lcc}
\minimize \,\,  \sum_{S \in \mathcal{C}} \left( \frac{1}{2}\, \cut(S) + \lambda {|S| \choose 2} \right),
\end{equation}
where $\mathcal{C}$ is a non-overlapping clustering of the nodes and $\cut(S)$ denotes the number of edges leaving a cluster $S\in \mathcal{C}$. The value of $\lambda$ controls the size and density of clusters formed by optimizing the objective. Furthermore, due to its relationship with LambdaCC~\cite{Veldt:2018:CCF:3178876.3186110} (Section~\ref{lprime-lcc} and Appendix~\ref{app:lamcc}), it is known that several other well-studied objective functions for graph clustering are captured as special cases of LambdaPrime for fixed values of $\lambda$. This includes relationships with sparsest cut, cluster deletion~\cite{ShamirSharanTsur2004}, and modularity clustering~\cite{newman2004modularity}. 

A number of other closely related graph clustering objective functions also rely on tunable resolution parameters~\cite{Arenas2008analysis,Delvenne2010stabilitypnas,ReichardtBornholdt2004,ReichardtBornholdt2006,traag2011narrow}. Many of these can be viewed as generalizations of the popular modularity clustering objective~\cite{newman2004modularity}. In this manuscript we broadly refer to these as \emph{parametric} graph clustering objective functions, where the parameter in question is specifically a resolution parameter that controls edge density within clusters. The characteristic that most distinguishes LambdaPrime from other parametric objectives is its equivalence with a weighted version of Correlation Clustering~\cite{Bansal2004correlation}, a problem that has been studied extensively from the perspective of approximation algorithms. The modularity objective is NP-hard to approximate to within any multiplicative factor~\cite{Dinh:2015:NCV:2919336.2920600}, and thus techniques typically used for optimizing generalizations of modularity are heuristics with no approximation guarantee~\cite{Arenas2008analysis,Delvenne2010stabilitypnas,genLouvain_software,ReichardtBornholdt2004,traag2011narrow,ReichardtBornholdt2006,traag2019leiden}. However, LambdaPrime corresponds to a linear transformation of the modularity objective with a resolution parameter, and permits several algorithmic guarantees based on applying algorithmic techniques developed for Correlation Clustering. In particular, one can obtain an $O(\log n)$ approximation for any value of $\lambda \in (0,1)$ using standard weighted Correlation Clustering algorithms~\cite{CharikarGuruswamiWirth2005,demain2006ccgen}. Several other linear programming based algorithms specifically for the related LambdaCC problem have been developed as well for different parameter regimes~\cite{gleich2018ccgen,Veldt:2018:CCF:3178876.3186110}. 

\subsection{Clustering in All Parameter Regimes}
Optimizing a parametric clustering objective over a wide range of resolution parameters is a useful approach for
detecting different types of clustering structure in the same network. This is standard practice in regularized
statistical fitting, another area with parameterized objectives. As mentioned, however, most previous approaches for parametric graph
clustering rely on applying heuristic clustering techniques with no
guarantees.
Meanwhile, although approximation algorithms for LambdaPrime can be directly derived from existing work on LambdaCC and
Correlation Clustering, these only provide approximation guarantees for a single fixed value of $\lambda$. In this
paper, we focus on finding \emph{families} of clusterings that come with rigorous optimality guarantees for an entire
range of parameter values in a parametric graph clustering objective. More precisely, we say that a family of clusterings solves (or approximates) a parametric objective \emph{in all parameter regimes} if, for every value of the resolution parameter, the family contains a solution (or approximate solution) to the objective. In our work, we seek families satisfying guarantees both in terms of the approximation factor, as well as in terms of the number of clusterings needed to attain such an approximation factor for all values of a resolution parameter. 



 \subsection{Overview of Contributions}
 In our work, we provide new lower bounds and techniques for exactly or approximately solving the LambdaPrime objective in all parameter regimes. We provide an outline of our paper, including informal statements of our main results and a discussion of their significance
 
 \paragraph*{Bounding the Size of Optimal Solution Families}
 We begin by proving a bound on the number of clusterings needed to \emph{optimally} solve LambdaPrime in all parameter regimes.
 \begin{theoremintro}
 	\label{thm:optimal2}
 	(Section~\ref{sec:optimal}, Theorem~\ref{thm:optimal})
 	Given a graph $G = (V,E)$, there exists a family of~$|E|$ or fewer clusterings which, for every value of $\lambda \in (0,1)$, contains an optimal LambdaPrime clustering for that $\lambda$. On star graphs, this bound is tight (Section~\ref{sec:star}).
 \end{theoremintro}
This theorem tells us that even though there are an exponential number of ways to cluster a graph, a linear number of clusterings is sufficient to characterize an optimal family of clusterings. Furthermore, since the theorem is specifically proven for optimal clusterings, the same result also holds for related parametric clustering objectives that correspond to linear transformations of LambdaPrime and its weighted variants, including LambdaCC~\cite{Veldt:2018:CCF:3178876.3186110}, the constant Potts model of Traag et al.~\cite{traag2011narrow}, a related Potts model of Reichardt and Bornholdt~\cite{ReichardtBornholdt2004}, and variants of modularity clustering with a resolution parameter~\cite{Arenas2008analysis,ReichardtBornholdt2006}. Finally, this theorem is equivalent to proving that the parametric integer linear program (ILP) corresponding to the LambdaPrime objective, a piecewise-linear function, has a linear number of \emph{breakpoints}, i.e., points where the slope changes. This is significant given that in general, parametric ILPs may contain an exponential number of breakpoints~\cite{Carstensen:1983:CPI:3119410.3119525}.

\paragraph*{Obtaining Approximate Solutions in All Parameter Regimes}
In practice it is NP-hard to solve LambdaPrime for even a single value of the resolution parameter. For a fixed value of $\lambda \in (0,1)$, one can obtain an approximately optimal solution by solving a linear programming relaxation for LambdaPrime and rounding it, with existing techniques~\cite{demain2006ccgen,gleich2018ccgen,Veldt:2018:CCF:3178876.3186110}.
 However, when $\lambda$ is treated as a varying parameter, the LP relaxation of LambdaPrime corresponds to a parametric linear program. In general, one may need an exponential number of feasible solutions to solve a parametric LP in all parameter regimes~\cite{Murty1980}. Despite the bound we prove on the number of breakpoints of the LambdaPrime ILP, this does not hold for the LP relaxation. We overcome this challenge by bounding the number of solutions needed to \emph{approximate} the LP in all regimes.
\begin{theoremintro}
	(Section~\ref{sec:log_lp_eval}, Theorems~\ref{thm:approx} and~\ref{cor:approxlamcc})
	For any $\varepsilon > 0$, there exists a (poly-time computable) family of $O\big(\frac{\log n}{\log (1+\varepsilon)}\big)$ feasible LP solutions which, for any value of $\lambda$, contains a $(1+\varepsilon)$-approximate solution to the LambdaPrime LP relaxation. These can be rounded to produce a family of clusterings that, for every $\lambda \in (0,1)$, contains an $O(\log n)$-approximate solution to the LambdaPrime objective. 
\end{theoremintro} 
We note that as $\varepsilon \rightarrow 0$, the number of clusterings in our computed family behaves as
$O(\frac{1}{\varepsilon} \log n)$. However, given that our aim is simply to round LambdaPrime LP solutions to produce
clusterings that are within $O(\log n)$ of optimal, it suffices to treat~$\varepsilon$ as a constant (e.g., $\varepsilon
= 1$ or larger). Thus, in practice, the size of the approximating family is~$O(\log n)$.

Our ability to approximate the LP relaxation of a clustering objective in all parameter regimes is useful even when LP solutions are not rounded to produce approximate clusterings. Solutions to the LP provide lower bounds for evaluating the performance of heuristic clustering techniques, and can also be useful for learning how to set graph clustering resolution parameters in practice~\cite{Veldt2019learning}. 

\paragraph*{Asymptotic Tightness of Results}
One of the central contributions of our work is a proof that our logarithmic upper bound for approximating the LP relaxation in all parameter regimes is in fact asymptotically tight.
\begin{theoremintro}
	\label{thm:logn2}
	(Section~\ref{sec:rings}, Theorem~\ref{thm:logn}.)
	For the class of ring graphs with $n = 2^k$ nodes ($k \in \mathbb{N}$), for every $\varepsilon > 0$, at least $\Omega(\log n)$ feasible LP solutions are needed in order to approximate the LP relaxation of LambdaPrime in all parameter regimes to within a factor $(1+\varepsilon)$.
\end{theoremintro}
The proof of this result relies on several connections between our parametric clustering problem and concave function approximation~\cite{magnati2004ipco}. The optimal solution curve for the LambdaPrime LP relaxation corresponds to an increasing, concave, and piecewise-linear function in terms of $\lambda$~\cite{Carstensen:1983:CPI:3119410.3119525}. Approximating the LP relaxation in all parameter regimes with a small number of feasible solutions is therefore equivalent to finding another piecewise-linear curve with a small number of linear pieces. Previously, Magnanti and Stratila~\cite{magnanti2012separable} demonstrated that in order to approximate the square root function $\mathit{sqrt}(x) = \sqrt{x}$ via a piecewise-linear upper bound over the interval $[a,b]$, at least $\Omega(\log \frac{b}{a})$ linear pieces are needed. Although this bound does not immediately imply any result for parametric clustering, we use this as a step in proving a similar lower bound for the LambdaPrime LP relaxation.

To show our full result, we prove several special properties satisfied by the LambdaPrime LP relaxation on ring graphs. We then prove a sequence of upper and lower bounds on the LP relaxation values, ultimately bounding it above and below in terms of the square root function. Finally, we demonstrate that these bounds, in conjunction with the results of Magnanti and Stratila~\cite{magnanti2012separable}, imply that $\Omega(\log n)$ feasible LP solutions are needed to approximate the LP relaxation in all parameter regimes. The bounds we show for the LP relaxation are loose enough that our lower bound is $\Omega(\log n)$ independent of the size of $\varepsilon$. Thus, as $\varepsilon \rightarrow 0$, the size of our approximate solution family behaves as $O(\frac{1}{\varepsilon} \log n)$, whereas our lower bound is still $\Omega(\log n)$. Nevertheless, if
we treat~$\varepsilon$ as a constant, which is the most natural choice for our setting, our upper and lower bounds
are asymptotically tight.

\paragraph*{Overcoming the Logarithmic Barrier}
 Although $\Omega(\log n)$ feasible solutions are required to approximate the LP relaxation in all parameter regimes in
the worst case, this is not the case for all graphs. We demonstrate that on star graphs, a single LP solution is
sufficient to optimize the LP relaxation in all non-trivial parameter regimes. Motivated by this observation, we develop
a new approach for more carefully selecting values of $\lambda$ at which to solve the LP relaxation. We adapt previous
techniques on sensitivity analysis in parametric linear programming~\cite{jansen1997sensitivity,nowozin2009solution} to
develop an approach for computing the exact range of $\lambda$ values for which a feasible LP solution is
\emph{approximately} optimal for the LambdaPrime LP. We hence find a family of LP solutions whose size is bounded in terms of the \emph{minimum} number of solutions needed to approximate the LP in all parameter regimes.
Thus, in certain cases this approach could release us from computing the LP relaxation at a logarithmic number of~$\lambda$
values.
Furthermore, by retroactively refining the family of solutions we obtain using our technique, we can prove the following result:
 \begin{theoremintro}
 (Section~\ref{sec:fe}, Theorem~\ref{thm:2approx}). 
 For $\varepsilon > 0$, let $M_\varepsilon$ be the minimum number of LP solutions needed to approximate the LambdaPrime LP relaxation in all parameter regimes to within a factor~$(1+\varepsilon)$. We obtain a family of $2 M_\varepsilon$ or fewer LP solutions that contains a $(1+\varepsilon)$-approximate solution to the LP in every parameter regime. 
 \end{theoremintro}
 We note that the above bound applies to the \emph{size} of the family of LP solutions we obtain using our techniques.
In practice, however, we will need to evaluate the LP relaxation more than $2M_\varepsilon$ times to actually obtain
this family. We also prove several results bounding the number of times we need to solve an LP, in terms of values which
depend on a minimum-size solution family. In the worst case, we may still need to evaluate the LP $O(\log n)$ times.
Nevertheless, this result shows that, without changing our worst-case asymptotic runtime, we can find a family of LP solutions that is nearly optimal in terms of output size. These can then be rounded to obtain a provably small family of clusterings that captures the full clustering structure of a given input graph.

\section{The LambdaPrime Objective}
LambdaPrime is an objective function for clustering graphs based on a tunable resolution parameter $\lambda \in (0,1)$. The objective seeks to minimize the number of edges crossing between clusters, subject to a regularization term that controls cluster size. Formally, for a graph $G = (V,E)$, resolution parameter $\lambda$, and a clustering $\mathcal{C}$, the LambdaPrime score for $\mathcal{C}$ is
\begin{equation}
\label{eq:lcc2}
\lcc(\mathcal{C},\lambda) = \sum_{S \in \mathcal{C}} \left( \frac{1}{2}\, \cut(S) + \lambda {|S| \choose 2} \right),
\end{equation}
where $S \in \mathcal{C}$ is used to denote an individual cluster in $\mathcal{C}$ and $\cut(S)$ is the number of edges that are incident on exactly one node in $S$. The LambdaPrime objective can also be expressed formally as an integer linear program (ILP):
\begin{equation}
\label{eq:lccilp}
\begin{array}{lll} \text{minimize }  \opt(\lambda) &= \sum_{(i,j) \in E} x_{ij} +\, &\sum_{i< j}\lambda (1-x_{ij}) \\ \subjectto  & x_{ij} \leq x_{ik} + x_{jk} & \text{ for all $i,j,k$} \\ & x_{ij} \in \{0,1\} & \text{ for all $i,j$.} \end{array}
\end{equation}
The variable $x_{ij}$ represents the binary \emph{distance} between nodes in a clustering. If $x_{ij} = 1$, nodes $i$ and $j$ are in different clusters, whereas $x_{ij} = 0$ indicates they are clustered together. In this way, clusterings of a graph are in one-to-one correspondence with feasible solutions to the ILP. Note that the number of clusters to form is not determined ahead of time, but is implicitly controlled by the parameter $\lambda$. 

\subsection{Correlation Clustering and LambdaCC}
\label{lprime-lcc}
LambdaPrime is equivalent to a special weighted instance of Correlation Clustering~\cite{Bansal2004correlation}. The latter problem clusters a dataset based on pairwise similarity and dissimilarity scores between data objects. These scores are typically modeled as signed and weighted edges between nodes in a graph. An instance of LambdaPrime corresponds specifically to a signed graph in which \emph{some} node pairs share a positive edge with weight one, and \emph{all} node pairs share a negative edge with weight $\lambda$. Placing a pair of nodes in the same cluster results in a penalty of $\lambda$, while separating nodes that have a positive edge results in a penalty of one. 

LambdaPrime is best viewed as a slight variant of a previously introduced framework for graph clustering called Lambda Correlation Clustering (LambdaCC)~\cite{Veldt:2018:CCF:3178876.3186110}. In the LambdaCC framework, an input graph $G = (V,E)$ is converted into a complete signed graph in which positive edges (corresponding to edges in $G$) have weight $(1-\lambda)$, and negative edges (corresponding to non-edges in $G$) have weight $\lambda$. Thus, each pair of nodes participates in a strictly positive or a strictly negative relationship. For any given $\lambda \in (0,1)$, LambdaCC and LambdaPrime share the same set of optimal solutions. However, these objectives differ by an additive term $\lambda |E|$, and thus differ in terms of approximations. Another difference is that the LambdaPrime solution value is strictly increasing with $\lambda$, though this is not the case for LambdaCC.

Throughout the manuscript, we focus on proving results for the LambdaPrime objective, given in framework~\eqref{eq:lcc}. This enables the clearest exposition of our techniques and results for parametric graph clustering, without changing the fundamental nature of the results. In Appendix~\ref{app:lamcc}, we discuss how to adapt our techniques to obtain nearly identical theorems and techniques for the LambdaCC objective. Though some results differ in terms of constant factors, for both objectives we can obtain a family of logarithmically many clusterings that, for every $\lambda \in (0,1)$, contains an approximation solution. The proof of asymptotic tightness given in Section~\ref{sec:rings} also holds up to differences in constant factors. In the Appendix we also discuss how our techniques can be applied to node-weighted variants of the LambdaPrime.

\paragraph*{Relation to other clustering objectives}
Despite its simplicity, LambdaPrime generalizes and interpolates between a number of previously studied objectives for graph clustering. In our previous work~\cite{Veldt:2018:CCF:3178876.3186110}, we showed that the optimal solutions of LambdaCC (and thus the optimal solutions of LambdaPrime) interpolate between solutions to the sparsest cut objective and the cluster deletion problem. The scaled sparsest cut of a graph $G = (V,E)$ is the bipartition $\{S, \bar{S} \}$ which solves the following objective:
\begin{equation}
\label{ssc}
\min_{S \subset V} \,\, \frac{\cut(S)}{|S| | \bar{S}|},
\end{equation}
where $\bar{S} = V \backslash S$ is the complement of a set $S \subset V$ and $\cut(S)$ equals the number of edges crossing from $S$ to $\bar{S}$. Cluster deletion is the problem of partitioning $G$ into cliques in a way that minimizes the number of edges between cliques. We formalize the relationship between LambdaPrime and these two objectives with a theorem, originally proven for the LambdaCC framework~\cite{Veldt:2018:CCF:3178876.3186110}.
\begin{theorem}
	\label{thm:lcc}
	Let $G = (V,E)$ be a graph, and define $\lambda^* = \min_{S \subset V} \cut(S)/ (|S| |\bar{S}|)$.
	\begin{itemize}
		\item For any $\lambda \leq \lambda^*$, placing all nodes in one cluster optimizes the LambdaPrime objective. 
		\item For any $\lambda \in (\lambda^*,1)$, the optimal LambdaPrime clustering will contain at least two clusters. There exists some $\lambda > \lambda^*$ such that the optimal LambdaPrime clustering will be the bipartition $\{S^*, \bar{S}^*\}$ which optimizes objective~\eqref{ssc}.
		\item For any $\lambda \in (|E|/ (1+ |E|),1)$, the optimal LambdaPrime clustering will be optimal for the cluster deletion problem. In other words, all clusters will be cliques, and the number of edges crossing between clusters will be minimized.
	\end{itemize}
\end{theorem}
Note that for any connected graph, $\lambda^* \geq 4/n^2$. This is tight for any graph that can be partitioned into two equal sized sets, with only a single edge crossing between the partitions. Thus, when searching for clusterings that optimize LambdaPrime, it suffices to consider $\lambda \in \left( \frac{4}{n^2}, 1 \right)$. 

 LambdaPrime is also related to other generalized clustering objectives that rely on tunable resolution parameters, including the constant Potts model~\cite{traag2011narrow}, and generalizations of the modularity objective that includes a resolution parameter~\cite{Arenas2008analysis,ReichardtBornholdt2006}. For an appropriate choice of parameter settings, these objectives are equivalent at optimality, though they differ in terms of approximations. 

\subsection{Approximations via Linear Programming}
Since LambdaPrime corresponds to a specially weighted variant of Correlation Clustering, we can obtain an $O(\log n)$ approximation guarantee for the objective for any value of the parameter~$\lambda$~\cite{demain2006ccgen}. This is obtained by solving and rounding the following LP relaxation:
\begin{equation}
\label{eq:lcclp}
\begin{array}{llll} \text{minimize } &  \lp(\lambda) = &\sum_{(i,j) \in E} x_{ij} \hspace{.2cm}+ &\sum_{i< j}\lambda (1-x_{ij}) \\ \subjectto  & & x_{ij} \leq x_{ik} + x_{jk} & \text{ for all $i,j,k$} \\ & & 0 \leq x_{ij} \leq 1 & \text{ for all $i,j$.} \end{array}
\end{equation}
Although better approximation guarantees exist for certain large values of~$\lambda$, in the worst case, the LP relaxation has an $\Omega(\log n)$ integrality gap, which can be shown by slightly adapting the integrality gap proof for LambdaCC~\cite{gleich2018ccgen}. In this paper, our goal is not to obtain new approximation guarantees for fixed values of $\lambda$. Instead, we show how to obtain approximately optimal solutions for \emph{all} values of the parameter using a small number of LP solves. When rounding LP solutions to produce approximately optimal clusterings, we default to considering the worst-case $O(\log n)$ rounding scheme, noting that in some parameter regimes, better guarantees are possible. 

\subsection{LambdaPrime and Parametric Programming}
If we do not treat $\lambda$ as a fixed value, the objective in problem~\eqref{eq:lccilp} corresponds to a parametric integer linear program in $\lambda$. We use $\opt(\lambda)$ to denote the optimal ILP score at a certain value of $\lambda$: this function $\opt$ is known to be be concave and piecewise-linear in $\lambda$~\cite{Carstensen:1983:CPI:3119410.3119525}. The \emph{breakpoints} of a parametric ILP are values of the parameter $\lambda$ at which a slope change occurs. In this context, a slope change corresponds to a parameter $\lambda$ at which the optimal clustering for LambdaPrime changes. Similarly,~\eqref{eq:lcclp} is a parametric linear program, whose solution we denote by $\lp(\lambda)$, and is also concave and piecewise linear in terms of $\lambda$. Breakpoints for the parametric LP are places at which the optimal feasible solution changes.

Previous work on parametric programming has shown that, in the worst case, parametric integer programs and parametric linear programs may have an exponential number of breakpoints~\cite{Carstensen:1983:CPI:3119410.3119525,Murty1980}. We will demonstrate that for the LambdaPrime parametric ILP, the number of breakpoints is in fact linear in terms of the number of edges, implying that a relatively small number of clusterings is able to capture all optimal LambdaPrime solutions. However, it remains NP-hard to find even one of these clusterings. Furthermore, for the LP relaxation it may not be the case that there are only a linear number of breakpoints, and it remains an open question whether the number of breakpoints is even polynomial. Instead, we demonstrate that the LP can be \emph{approximately} solved using only a logarithmic number of LP evaluations, which can then be rounded to approximately optimal solutions. 

\subsection{Concave Function Approximation}
Finding an optimal solution for either $\opt$ or $\lp$ at a single value of $\lambda$ corresponds to evaluating a function at a single point. Approximating either function over a range of $\lambda$ values is equivalent to approximating a concave, piecewise-linear function, using another concave and piecewise-linear function constructed from a set of clusterings (or feasible LP solutions in the case of the LP relaxation). 

Figure~\ref{fig:concave} displays the curves traced out by $\opt$ for a small synthetic graph. Each linear piece in the plot corresponds to a different clustering that remains optimal over a range of $\lambda$ values. In addition to being concave and piecewise linear, note that $\opt$ is strictly increasing in $\lambda$. Positive edges always have a fixed weight of one, and as $\lambda$ varies, making mistakes at negative edges becomes more expensive. Thus, as $\lambda$ increases, the objective value corresponding to every fixed clustering increases. Thus, for every graph, both the LambdaPrime objective~\eqref{eq:lccilp} and its LP relaxation~\eqref{eq:lcclp} are increasing functions. Due to the size and structure of the graph in Figure~\ref{fig:concave}, solutions to the LambdaPrime ILP and LP are in fact the same, i.e., $\opt(\lambda) = \lp(\lambda)$ for all $\lambda \in (0,1)$. Typically this will not be the case in practice. For larger graphs, it will be prohibitively expensive to compute ILP solutions, but solving the LP relaxation can still be accomplished in polynomial time. In recent work~\cite{veldt2019simods}, we showed how the linear programming relaxation of Correlation Clustering can be solved  in practice using memory-efficient projection methods.


\begin{figure}[t]
	\centering
	\begin{subfigure}[b]{0.2\linewidth}
		\includegraphics[width=\linewidth]{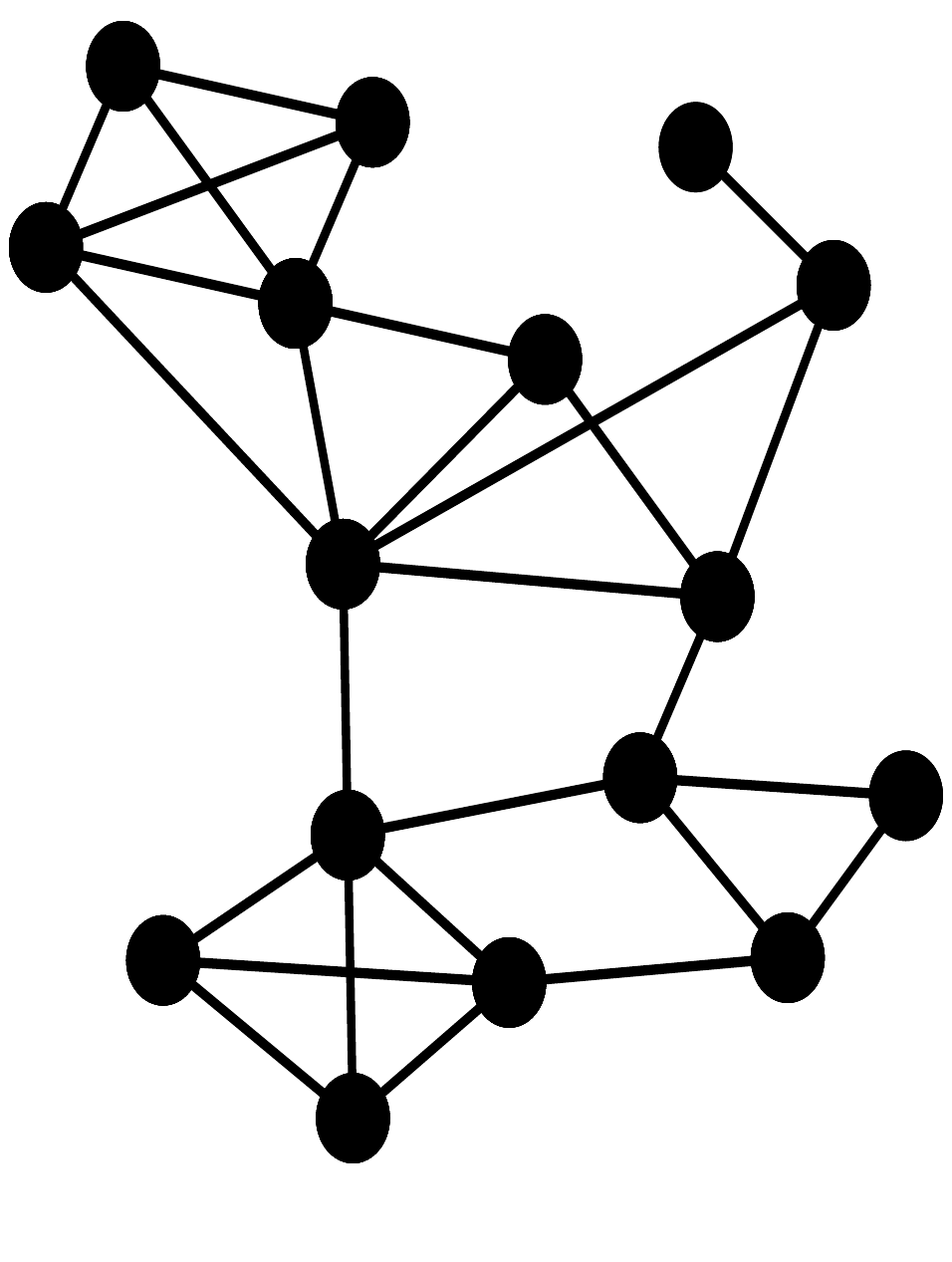}
		\caption{A small graph.}
	\end{subfigure}
\hfill
	\begin{subfigure}[b]{0.35\linewidth}
		\includegraphics[width=\linewidth]{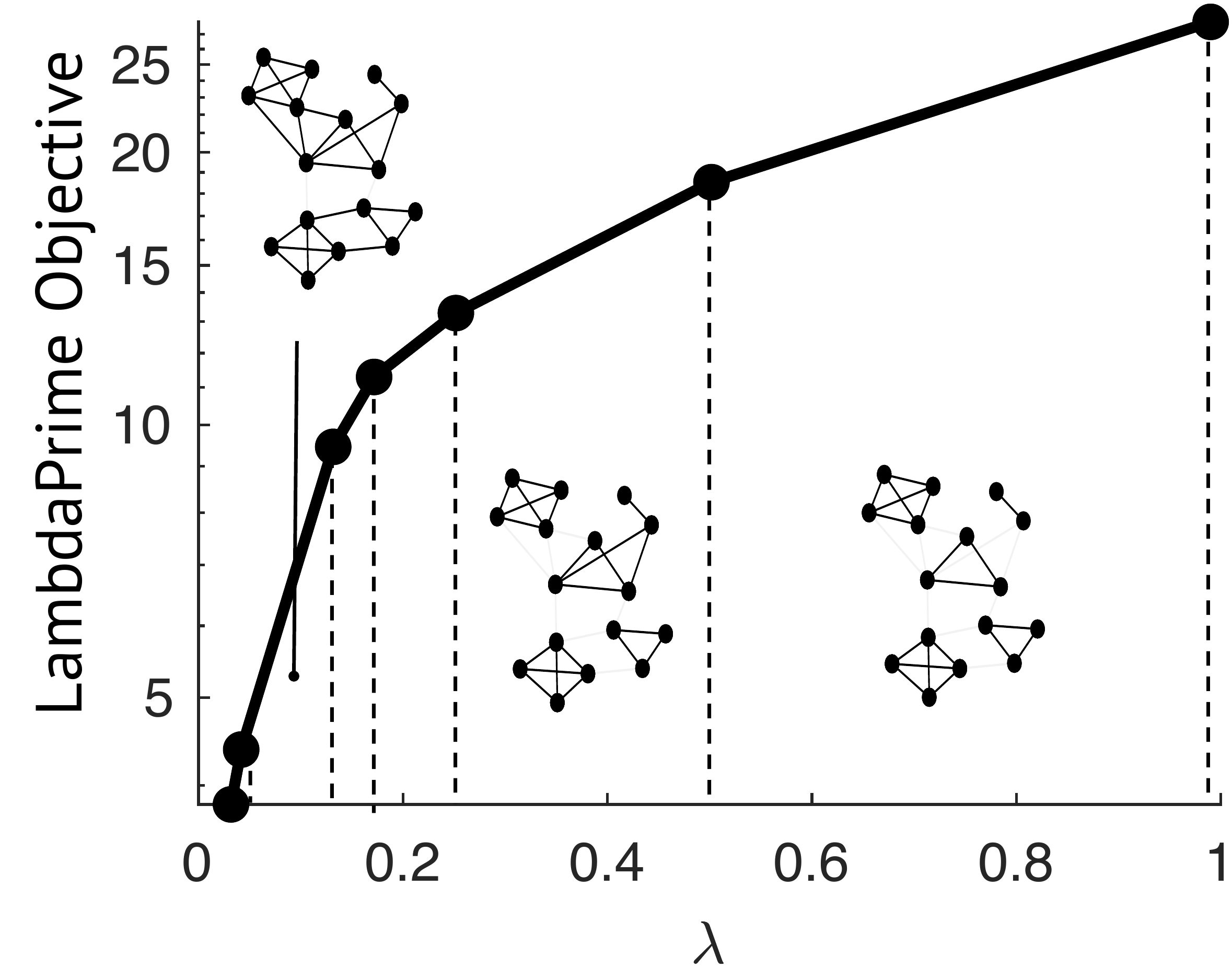}
		\caption{LambdaPrime solution curve.}
	\end{subfigure}
\hfill
	\begin{subfigure}[b]{0.35\linewidth}
		\includegraphics[width=\linewidth]{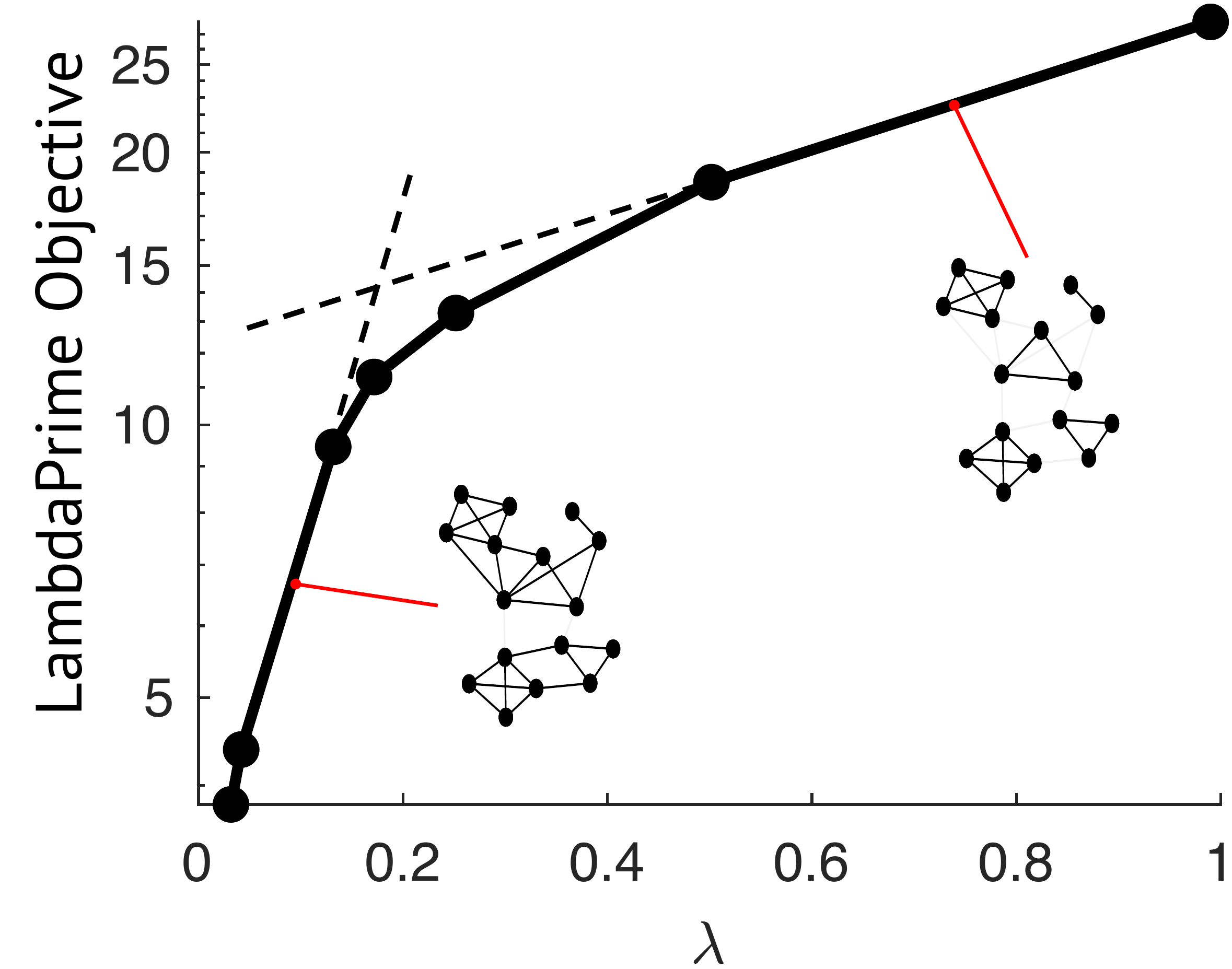}
		\caption{Clusterings $\rightarrow$ linear pieces.}
	\end{subfigure}
	\caption{The concave, piecewise-linear curve $\opt$ of optimal LambdaPrime solution scores for a small synthetic graph.}
	\label{fig:concave}
\end{figure}


Given any set of clusterings $\mathscr{C}$, we can define a new piecewise-linear function that approximates the LambdaPrime objective by identifying the clustering in $\mathscr{C}$ which best approximates LambdaPrime for a certain range of $\lambda$ values. In Figure~\ref{fig:concave} we illustrate this idea by extracting a sub-family of the optimal clusterings for the same small synthetic graph. The new approximate function has a smaller number of linear pieces, since we have selected a strict sub-family of clusterings, and upper bounds the function $\opt$. In general, the same principle holds for approximating the LambdaPrime LP relaxation for a certain graph in different parameter regimes. We will typically accomplish this by finding a set of feasible solutions, each of which \emph{exactly} minimizes the LP for some $\lambda$, and corresponds to one of the linear pieces of the function $\lp$. If this is done carefully, the resulting piecewise-linear curve will still remain a good approximation for $\lp$, despite containing far fewer linear pieces.


\section{Optimal LambdaPrime Clusterings}
\label{sec:optimal}
We begin by proving a bound on the number of clusterings needed to solve LambdaPrime in all parameter regimes.
\begin{theorem}
	\label{thm:optimal}
	Given any  graph $G = (V,E)$, let $c$ equal the minimum number of edges that must be removed in order to partition $G$ into cliques (i.e., $c$ is the cluster deletion solution). There exists a family of $c+1 \leq |E|$ or fewer clusterings, such that for every $\lambda \in (0,1)$, the family contains an optimal LambdaPrime clustering for that $\lambda$.
\end{theorem}
\begin{proof}
	As noted previously, LambdaPrime corresponds to a parametric ILP, whose solution curve $\opt$ is a concave, increasing, and piecewise-linear function in $\lambda \in (0,1)$. Let 
	\[0 < \lambda_1 < \lambda_2 < \cdots < \lambda_k < 1\]
	denote the breakpoints of the ILP, and use $\lambda_0 = 0$ and $\lambda_{k+1} = 1$ to denote the endpoints for our parameter space. For $t = 0, 1, 2, \hdots, k$, let $\textbf{x}^t = (x_{ij}^t)$ denote the feasible ILP solution that is optimal in the range $\lambda \in [\lambda_{t}, \lambda_{t+1}]$. Each $\textbf{x}^t$ encodes a clustering $\mathcal{C}_t$ of $G$ that is optimal in this range and corresponds to a linear piece of $\opt$. For each clustering $\mathcal{C}_t$, define
	\begin{align*}
	P_t = \sum_{(i,j) \in E} x_{ij}^t = \frac{1}{2}\sum_{S \in \mathcal{C}_t} \cut(S), \hspace{1cm}  N_t = \sum_{i<j} (1-x_{ij}^t) = \sum_{S \in \mathcal{C}_t} {|S| \choose 2},
	\end{align*}
	so that the LambdaPrime objective for an arbitrary $\lambda$ is
	\begin{align}
	\label{lamcct}
	\lcc(\mathcal{C}_t,\lambda) = P_t + \lambda N_t. 
	\end{align}
	From Theorem~\ref{thm:lcc}, we know that $\mathcal{C}_{k}$, which is optimal over $\lambda \in [\lambda_k, 1)$, will be an optimal solution for the cluster deletion objective. Thus, $P_k = c$. Note next that $P_t < P_{t+1}$ for $t = 0, 1, \hdots (k-1)$.
	To see why, note that $\mathcal{C}_t$ is optimal over $\lambda \in [\lambda_{t}, \lambda_{t+1}]$ and $\mathcal{C}_{t+1}$ is optimal for $\lambda \in [\lambda_{t+1}, \lambda_{t+2}]$. In particular, both clusterings are optimal at the breakpoint $\lambda_{t+1}$, and therefore
	\[P_{t} + \lambda_{t+1}N_t = P_{t+1} + \lambda_{t+1} N_{t+1}.\] 
	If $P_t = P_{t+1}$, then $N_t = N_{t+1}$, contradicting the fact that these clusterings are optimal for \emph{different} parameter ranges. If $P_t > P_{t+1}$, then $N_t < N_{t+1}$, which would imply that for $\lambda > \lambda_{t}$, $\mathcal{C}_t$ would be a better approximation than $\mathcal{C}_{t+1}$, another contradiction. Thus, $P_t < P_{t+1}$. Since the graph is unweighted, this means that the are at most $c +1 \leq |E|$ possible values for $P_t$ for $t = 0, 1, \hdots, k$, which includes the clustering in which all nodes are placed in a single cluster and $P_t = 0$. Thus, there are at most $c + 1$ clusterings in an optimal family.
\end{proof}
Given that LambdaPrime is NP-hard, we cannot hope to find families of optimal clusterings in practice. However, the above theorem tells us that although there are an exponential number of ways to cluster a graph, in practice only a small number of these clusterings are in fact needed to characterize the full clustering structure of one graph. Furthermore, in general, parametric integer linear programs may possess up to an exponential number of breakpoints~\cite{Carstensen:1983:CPI:3119410.3119525}. Our theorem shows that the LambdaPrime ILP avoids this worst-case scenario.

\subsection{Tightness on Star Graphs} 
\label{sec:star}
We end the section by observing that Theorem~\ref{thm:optimal} is in fact tight on star graphs. Consider an $n$-node star graph where node $1$ is the central node, and we refer to all other nodes as outer nodes. The optimal sparsest cut solution places one outer node with node 1, and all other nodes in a second cluster. For cluster deletion, the optimal solution places one outer node with node 1, and each other node in a singleton cluster, for a cluster deletion score of $c = n-2$. 

The LambdaPrime ILP interpolates between these solutions. The minimum scaled sparsest cut of the star graph is $\lambda^* = \frac{1}{n-1}$, so this will be the first breakpoint of the ILP. The final breakpoint is at $\lambda = \frac{1}{2}$, above which point it becomes suboptimal to make even a single negative edge mistake. Thus, when $\lambda \geq 1/2$, the cluster deletion solution is optimal. As $\lambda$ decreases from $\lambda = \frac{1}{2}$ to $\lambda = \frac{1}{n-1}$, the optimal solution will add outer nodes one by one to the cluster containing the central node. There will be exactly $c + 1 = n-1 = |E|$ such clusterings, counting down until all outer nodes have been merged with the central node. 

\section{Approximate LP Solutions in All Parameter Regimes}
\label{sec:approx}
Although the number of optimal ILP solutions needed to optimize the LambdaPrime objective exactly can be bounded above by $|E|$, the result does not hold for the LambdaPrime LP relaxation. In this section, however, we show that we can find a set of $O(\log n)$ feasible solutions to the LP relaxation which approximate the LP relaxation in every parameter regime. 

\subsection{Relation to Concave Function Approximation}
Before proceeding, we note a connection between our techniques and past results on concave function approximation. Recall that the LP relaxation~\eqref{eq:lcclp} is increasing, piecewise linear, and concave. Magnanti and Stratila~\cite{magnanti2012separable} showed how to approximate \emph{any} concave and increasing function $f: [a,b] \rightarrow \mathbb{R}^+$ (where $0 < a < b$), using a piecewise-linear function $p: [a,b] \rightarrow \mathbb{R}^+$. They accomplish this by evaluating $f$ at a logarithmic number of points between $a$ and $b$, and computing tangents at these points which can be joined to produce a piecewise-linear upper bound on $f$. In the context of solving the LambdaPrime LP relaxation, this strategy exactly corresponds to evaluating the LP at a logarithmic number of $\lambda$ values, and then using the resulting feasible solutions, each of which is optimal for some $\lambda$, as approximations for other nearby~$\lambda$. For the sake of completeness, below we provide self-contained proofs for approximating the LambdaPrime LP relaxation in all parameter regimes, specifically tailored to our graph clustering problem. 

\subsection{Logarithmic LP Evaluations}\label{sec:log_lp_eval}
Let $\varepsilon > 0$ be given and assume we wish to find a set of feasible solutions that provide a $(1+\varepsilon)$-approximate solution to the LambdaPrime LP relaxation for every $\lambda \in (4/n^2,1)$. Note that for $\lambda < 4/{n^2}$, the optimal clustering will place all nodes in a single cluster, so we do not need to consider LP solutions below this threshold. We first prove a lemma showing how well an optimal solution for the ILP (or LP) at one value of $\lambda$ approximates the ILP (respectively, the LP) when a nearby resolution parameter is used. 
\begin{lemma}
	\label{lem:approx}
Let $(x_{ij}^t)$ and $(x_{ij}^{t+1})$ be optimal solutions to the LambdaPrime ILP (respectively the LP relaxation) for resolution parameters $\lambda_t < \lambda_{t+1}$. Let $\delta = \frac{\lambda_{t+1}}{\lambda_t}$. Then $(x_{ij}^t)$ is a $\delta$-approximate solution for the LambdaPrime ILP (respectively, the LP relaxation) when~$\lambda_{t+1}$ is used, and $(x_{ij}^{t+1})$ is a $\delta$-approximate solution for the ILP (respectively, LP relaxation) when $\lambda_t$ is used.
\end{lemma}
\begin{proof}
	All steps of the proof hold regardless of whether we are optimizing over binary variables $x_{ij} \in \{0,1\}$ or relaxed distance scores $x_{ij} \in [0,1]$. For $k \in \{t, t+1 \}$, define
	\begin{align*}
	P_k = \sum_{(i,j) \in E} x_{ij}^k \text{ and } N_k = \sum_{i<j} (1-x_{ij}^k),
	\end{align*}
	so that the LambdaPrime score for $(x_{ij}^k)$ at an arbitrary value of $\lambda$ is $P_k + \lambda N_k$. 
	Since $(x_{ij}^t)$ and $(x_{ij}^{t+1})$ are optimal for their respective resolution parameters, and $\lambda_t < \lambda_{t+1}$, we have the following sequence of inequalities:
	\begin{align*}
	P_{t+1} + \lambda_t N_{t+1} \leq P_t + \lambda_{t+1} N_t  < \frac{\lambda_{t+1}}{\lambda_t} \left(P_t + \lambda_{t} N_t \right)<  \frac{\lambda_{t+1}}{\lambda_t} \left(P_{t+1} + \lambda_{t+1} N_{t+1}\right).
	\end{align*}
	Thus, both $(x_{ij}^t)$ and $(x_{ij}^{t+1})$ are at worst a $\delta$-approximation across the entire interval $[\lambda_t, \lambda_{t+1}]$, where $\delta = \lambda_{t+1}/\lambda_t$.
\end{proof}
We use Lemma~\ref{lem:approx} to construct a sequence of $\lambda$ values and corresponding optimal LP solutions (or ILP solutions if desired), to approximate the LambdaPrime objective in all parameter regimes.

\begin{theorem}
	\label{thm:approx}
	Let $\varepsilon > 0$ be given. There exists a set of $\lfloor \log_{1+\varepsilon} (n)\rfloor + 2$ feasible solutions to the LambdaPrime ILP (respectively, the LP relaxation), such that this set of solutions contains a $(1+\varepsilon)$-approximate solution to the ILP (respectively, the LP), for any $\lambda \in \big(4/{n^2}, 1\big)$.
\end{theorem}
\begin{proof}
	The proof is constructive. Set $\lambda_1 = 4/n^2$ and let $q = \floor*{\log_{(1+\varepsilon)^2} (n^2/4)}+1$. For $k = 2, 3, \hdots, q$, recursively define a sequence of $\lambda$ values by setting $\lambda_k = (1+\varepsilon)^2 \lambda_{k-1}$, and let $\lambda_{q+1} = 1/(1+\varepsilon)$. 
	Evaluate the LambdaPrime LP relaxation (or the ILP) at each of these $\lambda$ values to obtain solutions $(x_{ij}^1), (x_{ij}^2), \hdots , (x_{ij}^{q+1})$. By Lemma~\ref{lem:approx}, $(x_{ij}^1)$ is a $(1+\varepsilon)$-approximate solution for all $\lambda \in [\lambda_1, (1+\varepsilon)\lambda_1]$, $(x_{ij}^{q+1})$ is a $(1+\varepsilon)$-approximation for $\lambda \in [(1+\varepsilon)\lambda_q,1)$  and for any $k \in \{2, 3, \hdots, q\}$, $(x_{ij}^k)$ is a $(1+\varepsilon)$-approximate solution for all $\lambda \in \left[ (1+\varepsilon)\lambda_{k-1}, (1+\varepsilon) \lambda_k \right]$. Thus, using $q+1 < \lfloor 2 \log_{(1+\varepsilon)^2} (n)\rfloor + 2$ feasible solutions, we obtain a $(1+\varepsilon)$-approximate solution for every $\lambda \in [4/n^2,1)$.
\end{proof}	

While Theorem~\ref{thm:optimal} proved that there is a set of $c < |E|$ of fewer clusterings that contains an optimal LambdaPrime solution for every~$\lambda \in (0,1)$, Theorem~\ref{thm:approx} shows us that if we are content with $(1+\varepsilon)$-approximate solutions, a logarithmic number of clusterings suffices. Unfortunately, given that LambdaPrime is NP-hard, we cannot expect to find any optimal clusterings in practice. Nevertheless, Theorem~\ref{thm:approx} shows that we can, in polynomial time, obtain a set of clusterings which contain an \emph{approximately} optimal clustering for LambdaPrime in every parameter regime. 
We formalize this with a final theorem for the section.
\begin{theorem}
	\label{cor:approxlamcc}
	Given a graph $G = (V,E)$, we can obtain a set of $O(\log n)$ clusterings of $G$ which, for any $\lambda \in (0,1)$, contains a clustering that is within an $O(\log n)$ factor of optimal for the LambdaPrime objective. The procedure runs in $O(T_{LP} \log n)$ time, where $T_{LP}$ is the time it takes to solve the LambdaPrime LP relaxation for a single value of $\lambda$. 
\end{theorem}
\begin{proof}
	Fix a constant value $\eps > 0$ (e.g., $\eps$ = 1 will suffice) and solve the LP relaxation at the logarithmically spaced values of $\lambda$: $\lambda_1 < \lambda_2 < \cdots < \lambda_k$ as outlined in Theorem~\ref{thm:approx}. After this, we round each LP solution obtained using existing techniques~\cite{CharikarGuruswamiWirth2005}, which guarantees we have an $O(\log n)$-approximation for $\lambda_i$ for $i \in \{1, 2, \hdots, k\}$. Solving the LP relaxation is significantly more expensive than the rounding procedure, thus the overall runtime is $O(T_{LP} \log n)$.
	
	For each $\lambda_i$, let $\mathcal{C}_i$ denote the optimal LambdaPrime clustering, and let $\mathcal{\hat{C}}_i$ denote the clustering we obtain via LP rounding. For any clustering $\mathcal{C}$, define
	\begin{align*}
	P(\mathcal{C}) = \sum_{S \in \mathcal{C}} \frac{1}{2} \, \cut(S) \text{ and } N(\mathcal{C}) = \sum_{S \in \mathcal{C}} {|S| \choose 2},
	\end{align*}
	so that $\lcc(\mathcal{C},\lambda) = P(\mathcal{C}) + \lambda N(\mathcal{C})$ is the LambdaPrime objective for clustering $\mathcal{C}$ and resolution parameter $\lambda$. 
	
	Now consider an arbitrary $\lambda' \in (0,1)$, and let $\mathcal{C}'$ be optimal for this $\lambda'$. Note that there exists some $\lambda_i$ such that 
	\[\delta = \frac{\max(\lambda', \lambda_i)}{\min( \lambda',\lambda_i)} \leq (1 + \varepsilon)\,.\]
	Since our aim is to obtain $O(\log n)$-approximate solutions, the choice of $\varepsilon$ used will only change the approximation factor by a constant. 
	If we assume that $\lambda_i \leq \lambda'$, then
	\begin{align*}
	\lcc(\mchi, \lambda') &= P(\mchi) + \lambda' N(\mchi) \leq \frac{\lambda'}{\lambda_i} \Big( P(\mchi) + \lambda_i N(\mchi) \Big) \\
	& \leq \delta \cdot O(\log n)\cdot  \lcc(\mci, \lambda_i) \leq O(\log n) \cdot \lcc(\mathcal{C}', \lambda')\,, 
	\end{align*}
	where in the last step we have used the fact that the optimal LambdaPrime solution value is strictly increasing as $\lambda$ increases. On the other hand, if $\lambda' < \lambda_i$, then
%
	\begin{align*}
	\lcc(\mchi, \lambda') &= P(\mchi) + \lambda' N(\mchi) < \Big( P(\mchi) + \lambda_i N(\mchi) \Big) \\
	& \leq O(\log n)\cdot  \Big( P(\mci) + \lambda_i N(\mci) \Big) \leq O(\log n) \cdot  \Big( P(\mathcal{C}') + \lambda_i N(\mathcal{C}') \Big)\\
	& \leq O(\log n) \cdot \frac{\lambda_i}{\lambda'}\Big( P(\mathcal{C}') + \lambda'N(\mathcal{C}') \Big) \leq O(\log n)\cdot  \lcc(\mathcal{C}', \lambda').
	\end{align*}
Thus, the result is proven. 
\end{proof}
\section{Lower Bounds on Ring Graphs}
\label{sec:rings}
We now show that the number of LP solutions used to approximate the LambdaPrime relaxation in Theorem~\ref{thm:approx} is tight up to a constant factor. More precisely, there exists a class of \emph{ring} graphs for which we need at least $\Omega(\log n)$ feasible solutions to the LP relaxation in order to get a constant factor approximation for the LP in all parameter regimes. 

\subsection{Overview of Proof}
Our proof, in short, is to demonstrate that for a specific class of ring graphs, the LambdaPrime LP relaxation, as well as the original LambdaPrime objective, behave similarly enough to the square root function that we are able to adapt and apply the results of Magnanti and Stratila~\cite{magnanti2012separable} on concave function approximation. We first prove a new characterization of the optimal LP solution for ring graphs, and then demonstrate how this function can be bounded below and above in terms of a scaled version of the square root function. We then use these bounds to adapt the lower bound result of Magnanti and Stratila~\cite{magnanti2012separable} for the square root function, to obtain a similar lower bound for approximating the LambdaPrime LP relaxation.

\subsection{A Simple Class of Ring Graphs}
Let $G_k = (V,E)$ be a ring graph with $n = 2^k$. We will consider any $k \geq 3$, and always assume the nodes are ordered so that node $i$ is adjacent to node $i+1$ for $i = 1,2,  \hdots,  n-1$, and nodes $1$ and $n$ are adjacent. Throughout the section, we will explicitly make use of the fact that LambdaPrime corresponds to an instance of correlation clustering, defined on a signed graph. Specifically for the ring graph, every edge $(i,j) \in E$ is viewed as a positive edge $(i,j) \in E^+$ with weight one, and for every $(i,j) \in V \times V$ there is a negative edge $(i,j) \in E^-$ with weight $\lambda$. Figure~\ref{fig:g3} displays a picture of $G_3$, the smallest graph in this class.

	
	
Recall from Theorem~\ref{thm:lcc} that LambdaPrime interpolates between the scaled sparsest cut objective and the cluster deletion solution. The scaled sparsest cut for a ring graph is
\[ \lambda_1 = \min_{S \subset V} \frac{\cut(S)}{|S| |\bar{S}|} = \frac{2}{\frac{n}{2} \cdot \frac{n}{2}} = \frac{8}{n^2}. \]
The optimal cluster deletion solution on $G_k$ is to pair up each node with a single adjacent vertex. Thus, half the edges are cut, and half are not. This will be optimal for any value of~$\lambda \geq \frac{1}{2}$. We will therefore restrict our attention to LP solutions in the range $\left( \frac{8}{n^2}, \frac{1}{2} \right)$.	

	\begin{figure}[t]
		\centering
		\includegraphics[width=.5\linewidth]{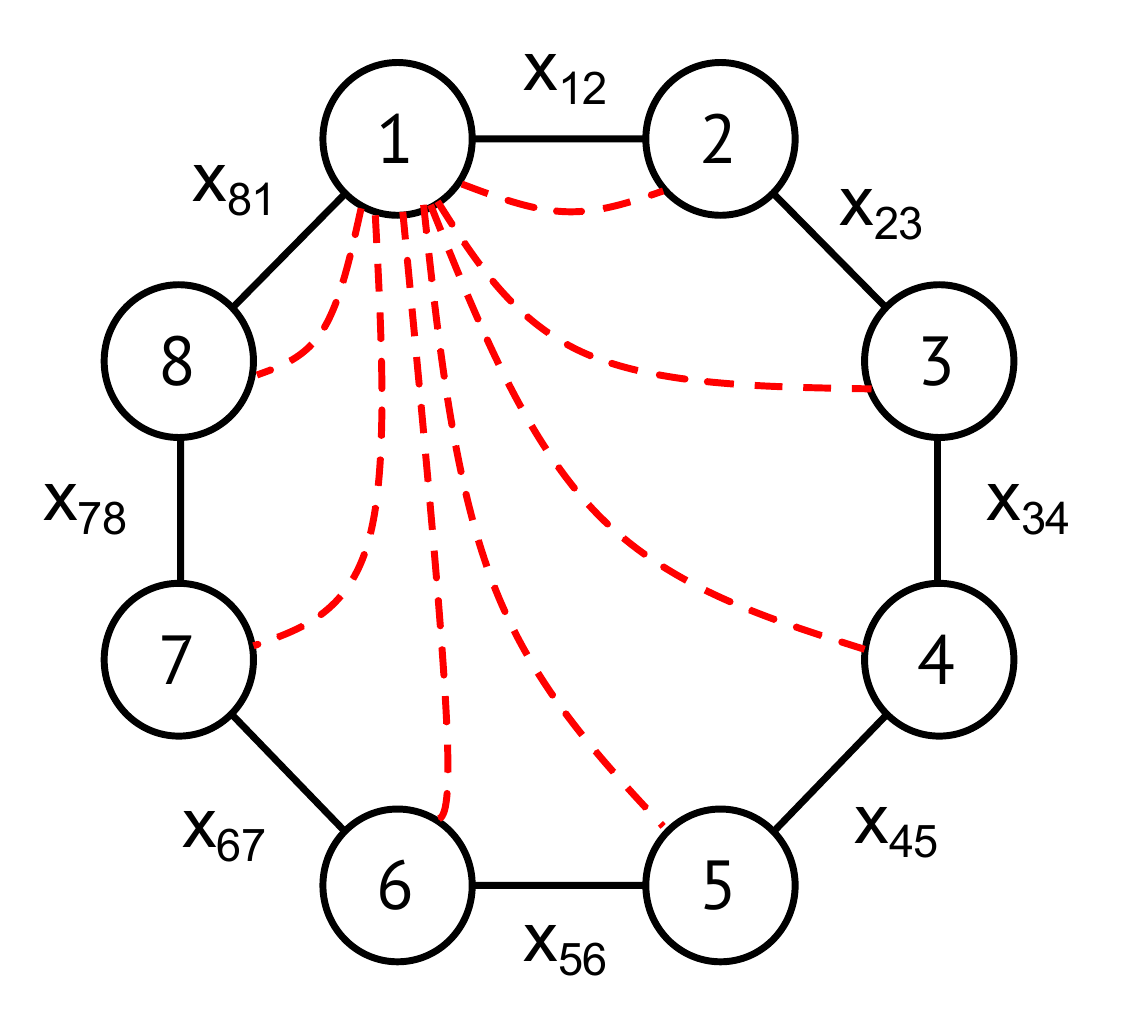}
		\caption{Ring graph $G_3$ with $n = 2^3$ nodes. All pairs of nodes share a negative edge of weight $\lambda$. We illustrate the negative edges adjacent to node 1 with red dashed lines, but omit other negative edges to simplify the illustration. An important part of understanding the LP relaxation of ring graphs in this class is to figure out the value of $x_{i,i+1}$, the LP distance score between nodes sharing a positive edge. Due to symmetry in the graph, this score will be the same for every value of $i$.}
		\label{fig:g3}
\end{figure}

\subsection{LambdaPrime LP on Ring Graphs}
In order to prove results regarding the LambdaPrime LP relaxation on this class of ring graphs, we will consider an alternative LP relaxation, that was originally considered by Wirth~\cite{wirth2004approximation} for the unweighted version of Correlation Clustering.
\begin{equation}
\label{neppc}
\begin{array}{lll} \text{minimize} & \sum_{(i,j)\in E} (1-\lambda)x_{ij} \,\,+ & \sum_{(i,j) \notin E} \lambda (1-x_{ij})\\ \subjectto
& x_{i_1, i_m}
\leq
\sum_{j=1}^{m-1} x_{i_j, i_{j+1}}
&\text{ for all $\NEPPC(i_1,i_2, \hdots , i_m)$} \\
& x_{ij}
\leq
1
& \text{ for all $(i,j)\notin E$} \\
& 0
\leq
x_{ij}
&\text{ for all $(i,j)$}\,.
\end{array}
\end{equation}
Each constraint in LP~\eqref{neppc} corresponds to a \textit{Negative Edge with Positive Path Cycle} (NEPPC), where $\NEPPC(i_1,i_2, \hdots , i_m)$ represents a sequence (i.e., a \emph{path}) of (positive) edges, 
\[ \{(i_1,i_2), (i_2, i_3), \hdots ,(i_{m-1},i_m)\} \subset E\,,\]
with a single non-edge (i.e., negative edge) completing the cycle: $(i_1, i_m) \in E^-$. Wirth~\cite{wirth2004approximation} proved that the set of optimal solutions to the NEPPC linear program~\eqref{neppc} is exactly the same as the optimal solution set to the canonical LP. Using the NEPPC LP relaxation, we prove a new way to express the optimal value of the LP solution on ring graphs for any $\lambda \in [8/n^2, 1/2]$. First, we make several observations about the NEPPC objective as it pertains to ring graphs.
\begin{itemize}
	\item There are two NEPPC constraints for each negative edge: one that comes from traveling clockwise around the ring graph (see Figure~\ref{fig:g3}), and the other from traveling counterclockwise. 
	\item As observed by Wirth~\cite{wirth2004approximation}, every negative edge $(i,j) \in E^-$ is either involved in a tight NEPPC constraint, or $x_{ij} = 1$. 
	\item If we assign $x_{i,i+1} = c$ for all $i = 1, 2, \hdots , n$ for some constant $c$, then for this fixed assignment, the LP will be minimized if $x_{ij} = \min \{ 1, c \cdot \mathit{dist}(i,j)\}$ for each $(i,j) \in E^-$, where $\mathit{dist}(i,j)$ is the shortest path distance (the number of positive edges) between nodes $i$ and $j$ in the ring graph. 
\end{itemize}
Using these observations, we prove the following characterization of LP solutions on the ring graph.
\begin{theorem}
	\label{thm:ringlp}
	For any $\lambda \in \left[  \frac{8}{n^2}, \frac12 \right]$ the optimal value of the LambdaPrime LP relaxation for the ring graph is
	\begin{equation}
	\label{eq:lp}
	\lp(\lambda) = \min_{t \in \mathbb{N} }  \frac{n}{t} \left( 1 + \lambda {t \choose 2} \right). \end{equation}
\end{theorem}
\begin{proof}
	We break the proof up into three parts:
	\begin{enumerate}
		\item We show that there exists an optimal solution $(x_{ij})$ in which all positive edge distances are the same. More precisely, there exists some $c \geq 0$ such, that $x_{i,i+1} = c$ for all $i = 1, 2, \hdots, n$.
		\item We prove that $\lambda \geq 8/n^2$ implies that $c > 0$ and $\frac{1}{c} \leq \frac{n}{2}$. 
		\item We prove that $\frac{1}{c}$ is in fact an integer. 
	\end{enumerate}
	Once we have proven the above steps, we will be able to re-express the LP relaxation in the form given by~\eqref{eq:lp}. Note that throughout the proof we will abuse notation slightly by assuming all subscripts follow modular arithmetic modulo $n$. For example, we express the positive edge LP distances as $x_{i,i+1}$ for $i = 1, \hdots, n$, with the understanding that due to modular arithmetic, $x_{n,n+1} = x_{n, 1} = x_{1,n}$.
	
	\paragraph*{Step 1: All positive distances are equal.}
	Let $\textbf{x}^1 = (x^{1}_{ij})$ be an arbitrary solution to the NEPPC LP relaxation~\eqref{neppc} for a fixed $\lambda$. Construct $n-1$ other optimal solutions $\textbf{x}^2,  \textbf{x}^3,  \hdots , \textbf{x}^{n}$  by setting $x_{ij}^t = x^1_{i+t, j+t}$ for all $i < j$ and $t =  2,3, \hdots, n$. In other words, we take advantage of the symmetry in the ring graph and ``rotate'' LP distance scores around the ring one node at a time to produce new optimal solutions. Then form 
	\[ \textbf{x}^* = \frac{1}{n} \sum_{j = 1}^n \textbf{x}^j.\]
	Note that $\textbf{x}^*$ is a convex combination of optimal LP solutions, and therefore is itself an optimal LP solution. Furthermore, for every $i = 1, 2, \hdots, n$, 
	\begin{align*}
	x_{i,i+1}^* &= \frac{1}{n} \left(x^1_{i,i+1} + x^2_{i,i+1} + \hdots + x^{n}_{i,i+1}\right) = \frac{1}{n} \left(x^1_{i,i+1} + x^1_{i+2,i+3} + \hdots + x^{1}_{i+n,i+n+1}\right).
	\end{align*}
	Regardless of the value of $i$, this is equal to the sum of all the positive edge distances in the LP solution $\textbf{x}^1$. Therefore, all positive edge distances in $\textbf{x}^*$ equal some constant $c \geq 0$. Note that this implies that $x^*_{ij} = \min \{ 1, c \cdot \mathit{dist}(i,j)\}$ for each $(i,j) \notin E$. 
	
	For the rest of the proof, we restrict our attention to LP feasible solutions that take this form. We will use $\lp(\lambda, x)$ to denote the value of the LP relaxation for the given value of $\lambda$, with all positive edges having distance $x$. We have shown that
	\[\lp(\lambda) = \min_{x\geq 0} \lp(\lambda,x),\]
	 i.e., for any $\lambda$, the overall minimum value of the LP is obtained by choosing the best value for positive edge length.
	
	\paragraph*{Step 2: Bounding $c$ from below.}
	Next we show that when $\lambda \geq 8/n^2$, the constant $c$ satisfies $\frac{1}{c} \leq \frac{n}{2}$.
	
	Assume that $c \in [0,2/{n})$. Then the shortest (positive) path distance between any two nodes is less than or equal to $n/2$, i.e., $x_{ij} = c \cdot \mathit{dist}(i,j) < 1$, for every pair $i < j$. The value of $\lp(\lambda,c)$ can be expressed succinctly by considering the LP cost associated at node~1, multiplying by the number of nodes~$n$, and then dividing by 2 since each (positive and negative) edge will be counted twice. Node 1 participates in two positive edges, $(1,n)$ and $(1,2)$, with LP costs of $c$:
	\[ (\text{LP cost of 2 positive edges}) = 2c.\]
	For each distance $t = 1, 2, \hdots, \frac{n}{2} -1$, node 1 shares in two negative edges with nodes at distance $t$, for an LP cost of:
	\begin{align*}
	\Big(\text{LP cost for $i$ s.t.\ } \mathit{dist}(1,i) < \frac{n}{2}\Big)  &= 2 \lambda \sum_{i = 1}^{n/2 -1} (1- ci) 
	= 2 \lambda\left( \frac{n}{2} - 1 - \frac{cn^2}{8} + \frac{cn}{4} \right).
	\end{align*}
	Finally, node 1 shares in a negative edge with a single node, with node ID ($\frac{n}{2} + 1$), at distance exactly $n/2$. Putting these together with their corresponding LP costs, we find that if $x_{i,i+1} = c$ for all $i = 1, 2, \hdots , n$, the LP score can be expressed as:
	\begin{align*}
	\lp(\lambda, c)  &= \frac{n}{2} \left(2c + 2 \lambda \sum_{i = 1}^{n/2 -1} (1- ci) + \lambda\left(1 - \frac{nc}{2}\right)\right)\\
	&= n \left(c +  \lambda \left( \frac{n}{2} - 1 - \frac{cn^2}{8} +\frac{cn}{4} \right) +  \frac{\lambda}{2} \left(1 - \frac{nc}{2}\right) \right)\\
	&= nc \left(1 - \frac{\lambda n^2}{8}\right) +  \lambda {n \choose 2}.
	\end{align*}
	Recall that we have assumed $c < 2/n$, and that $\lambda \geq 8/n^2$. If $c = 0$, the LP cost is equal to $\lambda {n \choose 2}$. Therefore, for any $c \in (0, 2/n)$,
	\begin{align*}
	\lp(\lambda, 0) =  \lambda {n \choose 2} &\geq\lp(\lambda, c) = nc \left(1 - \lambda\frac{n^2}{8}\right) +  \lambda {n \choose 2}\\
	&\geq\lp\left(\lambda, \frac{2}{n} \right) = 2 \left(1 - \lambda\frac{n^2}{8}\right) +  \lambda {n \choose 2}.
	\end{align*}
	In other words, when $\lambda > 8/n^2$, we obtain a smaller (i.e., better) LP score if we set all positive edge distances to be $x_{i,i+1} = 2/n$, rather than $x_{i,i+1} = c$ for any $c < 2/n$. If $\lambda = 8/n^2$, then setting $x_{i,i+1} = 2/n$ yields the same result as setting $x_{i,i+1} = c$ for any $c < 2/n$. We conclude that there exists a solution to the LambdaPrime relaxation in which $x_{i,i+1} = c$ for some $c \geq 2/n$ whenever $\lambda \geq 8/n^2$. An important consequence is that $x_{i, i+ n/2} = 1$ for every $i = 1, 2, \hdots ,n$. Thus, we will not incur any LP cost for negative edges between nodes that are at distance $n/2$ from each other. We can then express the optimal LambdaPrime LP relaxation score as $\lp(\lambda) = \min_{x \geq 2/n} \lp(\lambda, x)$ where
	\begin{align}
	\label{LPx}
	\lp(\lambda, x) = n \left(x + \lambda \sum_{i = 1}^{\floor*{\frac{1}{x}}} (1 - xi)   \right) = n \left(x + \lambda \floor*{\frac{1}{x}} - \frac{\lambda x}{2} \floor*{\frac{1}{x}} \left( \floor*{\frac{1}{x}} + 1\right) \right).
	\end{align} 

	\paragraph*{Step 3: Proving integrality of $1/c$.}
	Assume that $\lp(\lambda,x)$ is minimized over all $x \geq 2/n$ by some $c$ such that $1/c$ is \emph{not} an integer. Let $t = \floor*{\frac{1}{c}}$, which must be less than $n/2$ since $c > 2/n$. Select two new values $c_*$ and $c^*$ satisfying:
	\[ \frac{2}{n} <\frac{1}{t} = c_* < c < c^* < \frac{1}{t+1}.\]
	The values $c_*$ and $c^*$ are chosen so that they bound $c$ above and below, but applying the floor function to their reciprocal yields the same integer $t$. Therefore:
	\begin{equation}
	\label{newlp}
	\lp(\lambda, c) =n \left(c + \lambda t - \frac{\lambda c}{2} t \left( t+ 1\right) \right) =nc \left( 1 - \lambda \frac{t(t+1)}{2} \right) + \lambda n t.
	\end{equation}
	Observe finally that very similar expressions can be obtained for $\lp(\lambda, c^*)$ and $\lp(\lambda, c_*) = \lp(\lambda, 1/t)$. If $1 - \lambda t(t+1)/2 < 0$, then $\lp(\lambda, c^*) < \lp(\lambda,c)$. If $1 - \lambda t(t+1)/2 > 0$, then $\lp(\lambda, 1/t) < \lp(\lambda,c)$. Thus, for either of these two cases, it is strictly better to either increase or decrease the value of $c$, which would contradict the optimality of the LP relaxation when $c$ is chosen as the positive edge distance. The only remaining possibility is that $1 - \lambda t(t+1)/2  = 0$, but in this case we would obtain the same LP score by setting $x_{i,i+1} = \frac{1}{t}$ for all $i = 1, 2, \hdots ,n$. 
	We conclude therefore that there exists a solution in which the optimal positive edge distance is the reciprocal of an integer. We can adapt~\eqref{newlp} to see that the LambdaPrime LP relaxation can be expressed as
	\begin{equation}
	\lp(\lambda) = \min_{t \in \mathbb{N}} \, n \left( \frac{1}{t} + \lambda t - \lambda \frac{(t+1)}{2}\right) =  \min_{t \in \mathbb{N}} \, \frac{n}{t} \left( 1 + \lambda {t \choose 2} \right) .
	\end{equation}
\end{proof}

\subsection{Lower Bounding the LP Relaxation}
Next we define a lower bound on the linear programming relaxation function given in~\eqref{eq:lp} that is tight for certain choices of $\lambda$. 
\begin{lemma}
	\label{lem:g}
	Let $g: [0,1] \rightarrow \mathbb{R}_{+}$ be defined by
	\begin{equation}
	\label{eq:g} g(\lambda) = n \left( \sqrt{2\lambda} - \frac{\lambda }{2}\right).
	\end{equation}
	Then $g(\lambda) \leq \lp(\lambda) $ for every $\lambda \in [0,1]$. 
\end{lemma}
\begin{proof}
	For a fixed $\lambda$ consider the function
	\[ h(t) = \frac{n}{t} \left( 1 + \lambda {t \choose 2} \right) = \frac{n}{t} + \frac{\lambda n (t-1)}{2} \, .\]
%
	It is easy to check that $t = \sqrt{2}/\sqrt{\lambda}$ is a minimizer for the function $h$ over all $t \in \mathbb{R}_+$. Thus,
	\begin{align*}
	\min_{t \in \mathbb{R}_+} h(t) &= h (\sqrt{2}/\sqrt{\lambda})
	= n \sqrt{2\lambda} - \frac{\lambda n}{2} = g(\lambda).
	\end{align*}
	This shows that
	\begin{equation}
	\label{eq:fun}
	g(\lambda) =  \min_{t \in \mathbb{R}_+ }  \frac{n}{t} \left( 1 + \lambda {t \choose 2} \right) \leq \min_{t \in \mathbb{N}}  \frac{n}{t} \left( 1 + \lambda {t \choose 2} \right) = \lp(\lambda).
	\end{equation}	
\end{proof}

\subsection{Upper Bounding $\opt(\lambda)$ and $\lp(\lambda)$}\label{sec:upper_bound_opt_lp}
In general we know that $g(\lambda) \leq \lp(\lambda) \leq \opt(\lambda)$, where $\opt(\lambda)$ denotes the optimal LambdaPrime clustering score. In this section, we will show that the LambdaPrime objective and its LP relaxation can also be bounded above in terms of $g$. 

\subsubsection*{Special Values of $\lambda$ on the Ring}
\label{speciallam}
We begin by outlining a sequence of logarithmically spaced values of $\lambda \in \left[8/n^2,1/2\right]$ at which $g(\lambda) = \lp(\lambda) = \opt(\lambda)$ for class of ring graphs $G_k$. Let 
\[\lambda_1 = \frac{8}{n^2} = \frac{8}{2^{2k}} = \frac{2}{2^{2(k-1)}},\]
and more generally define
\begin{equation}
\label{eq:speciallam}
\lambda_{i} = \frac{2}{2^{2(k-i)}} \text{ for $i = 1, 2, \hdots, (k-1)$.}
\end{equation}
Note that for this sequence of $\lambda$ values, the value of the parameter $t$ that minimizes $h(t) = \frac{n}{t}\left(1 + \lambda {t \choose 2} \right)$ over all $t \geq 0$, is in fact an integer $t_i$:
\begin{align}
\label{opt-t}
t_i = \frac{\sqrt{2}}{\sqrt{\lambda_i}} = \frac{\sqrt{2 ( 2^{2k - 2i})}}{\sqrt{2}} = 2^{k-i}.
\end{align}
Thus $g(\lambda_i) = \lp(\lambda_i)$, as can be seen from~\eqref{eq:fun}. Furthermore, when $t = 2^{k-i}$, we can obtain a LambdaPrime score equal to $g(\lambda_i) = \lp(\lambda_i)$ by separating nodes into $\frac{n}{t_i}$ clusters made up of a path of $t_i$ nodes each, so we also have $\opt(\lambda_i) = g(\lambda_i)$.

\subsubsection*{A Useful Piecewise Linear Approximation}
Evaluating the LP relaxation at $\lambda_\ell$ for $\ell = 1, 2, \hdots ,(k-1)$ produces feasible LP solutions $(x_{ij}^1), (x_{ij}^2), \hdots, (x_{ij}^{k-1})$ that are optimal for $\lambda_1, \lambda_2, \hdots, \lambda_{k-1}$, respectively. As noted previously, the optimal ILP solutions also optimize the LP relaxation at these values of $\lambda$, so we will assume without loss of generality that $(x_{ij}^\ell)$ is a binary vector and encodes a clustering. Not only is $(x_{ij}^\ell)$ optimal for the resolution parameter $\lambda_\ell$, but it will also be a good approximate solution for the LP (or ILP) for other nearby values of $\lambda$. This fact can be used to define a new function that bounds $\lp$ and $\opt$ from above, while still being a good approximation.

More precisely, define a function $f: [\frac{8}{n^2}, \frac{1}{2}] \rightarrow \mathbb{R}^+$ so that $f(\lambda_i) = \lp(\lambda_i)$ for~$i = 1, 2, \hdots, (k-1)$. For~$\lambda \in [\lambda_i, \lambda_{i+1})$, the function is defined as follows for $i = 1, 2, \hdots , k-2$,
\begin{equation}
\label{ffunction}
f(\lambda) = \begin{cases}
\lp(\lambda, \frac{1}{t_i}), & \text{ if $\lambda \in [\lambda_i, 2\lambda_i = \frac{1}{2}\lambda_{i+1})$ }\\
\lp(\lambda, \frac{1}{t_{i+1}}), & \text{ if $\lambda \in [\frac{1}{2}\lambda_{i+1}, \lambda_{i+1} ).$}
\end{cases}
\end{equation}
where $\lp(\lambda, x)$ is defined in~\eqref{LPx}, and we recall that $\lp(\lambda,x) \geq \lp(\lambda)$. Notice then that we have the following set of inequalities for all $\lambda \in [\lambda_1, \lambda_{k-1}]$:
\[ g(\lambda) \leq \lp(\lambda) \leq \opt(\lambda) \leq f(\lambda)\,.\]
Intuitively speaking, the function $f$ corresponds to a simple strategy for approximating the LambdaPrime objective (and simultaneously, the LP relaxation), which is similar to the approach used in Section~\ref{sec:approx}. First it obtains a set of solutions $\textbf{X} = \{(x_{ij}^1), (x_{ij}^2), \hdots, (x_{ij}^{k-1}) \}$ that optimize both the LambdaPrime ILP and the LP relaxation at $\lambda_1, \lambda_2, \hdots, \lambda_{k-1}$, where~$\lambda_i$ is defined as in~\eqref{eq:speciallam}. For any other value of $\lambda$, the function $f$ reports the LambdaPrime score corresponding to the feasible solution $(x_{ij}^\ell) \in \textbf{X}$ that best approximates the objective at that $\lambda$. The result is a piecewise-linear upper bound on $\lp$ and $\opt$ that remains a good approximation for both.
The main result of this subsection is obtaining an upper bound on $\lp$ and $\opt$ via the function $f$.
\begin{lemma}
	\label{lemma:gbounds}
	For every $\lambda \in \left[\frac{8}{n^2},\frac12\right]$, for the class of ring graphs $G_k$, we have
	\begin{equation}
	\label{eq:gbounds2}
	g(\lambda) \leq \lp(\lambda) \leq \opt(\lambda) \leq \sqrt{2} g(\lambda).
	\end{equation} 
\end{lemma}
\begin{proof}
	Consider any $\lambda \in [\frac{8}{n^2}, \frac{1}{2}]$. We know that $\lambda \in [\lambda_i, \lambda_{i+1}]$ for some $i \in \{1,2, \hdots, k-2\}$, where the $\lambda_i$ are values of the parameter at which the optimal solutions to $g$, $\lp$, and $\opt$ coincide (Section~\ref{speciallam}). Remember also that for each $\lambda_i$ we had also defined $t_i = \sqrt{2}/\sqrt{\lambda_i}$.
	For any value of $\lambda$ we consider, we know $f(\lambda) = \lp(\lambda, {1}/{t_i})$ for some $i \in \{1, 2, \hdots ,(k-1)\}$. We will show that $f(\lambda) \leq \sqrt{2} g(\lambda)$ by considering two cases. 
	
	\paragraph*{Case 1: $\lambda \in [\lambda_i, 2\lambda_i)$.}
	In this case, by construction, the function $f$ approximates $g$ and $\lp$ by using $\frac{1}{t_i}$ as the distance between nodes in the LP, where $t_i$ is defined in~\eqref{opt-t}. Plugging $\frac{1}{t_i}$ into the expression for the LP score given in~\eqref{LPx}, we get
	\begin{align*}
	f(\lambda) &= \frac{n}{t_i} \left( 1 + \lambda {t_i \choose 2} \right)  = \frac{n\sqrt{\lambda_i}}{\sqrt{2}}  + \frac{n\lambda \sqrt{2}}{2 \sqrt{\lambda_i}} - \frac{n\lambda}{2}\\
	& \leq \frac{n\sqrt{\lambda}}{\sqrt{2}}  + \frac{n\lambda \sqrt{2}}{\sqrt{2} \sqrt{\lambda}} - \frac{n\lambda}{2}\leq \sqrt{2}\left(\frac{n\sqrt{\lambda}}{\sqrt{2}} - \frac{n\lambda}{2} + \frac{n\sqrt{\lambda}}{\sqrt{2}} \right) = \sqrt{2}g(\lambda).
	\end{align*}
	In the above, we have used the fact that since $\lambda/2 < 1$, we have $\sqrt{\lambda}/{\sqrt{2}} - \lambda / 2 > 0$.
	
	\paragraph*{Case 2: $\lambda \in \left[\frac{\lambda_{i+1}}{2}, \lambda_{i+1}\right)$.}
	To make notation slightly easier, we set $j = i+1$, and prove the result for $\big[\frac{\lambda_{j}}{2}, \lambda_{j}\big)$. For this case, we instead use the fact that $\frac{\lambda_j}{2} \leq \lambda \leq \lambda_j$, and we can show a similar set of steps to reach the result.
	\begin{align*}
	f(\lambda) &= \frac{n}{t_j} \left( 1 + \lambda {t_j \choose 2} \right) = \frac{n\sqrt{\lambda_j}}{\sqrt{2}}  + \frac{n\lambda \sqrt{2}}{2 \sqrt{\lambda_j}} - \frac{n\lambda}{2}\\
	& \leq \frac{n\sqrt{2}\sqrt{\lambda}}{\sqrt{2}}  + \frac{n\lambda \sqrt{2}}{2 \sqrt{\lambda}} - \frac{n\lambda}{2} \leq \sqrt{2}\left(\frac{n\sqrt{\lambda}}{\sqrt{2}} - \frac{n\lambda}{2} + \frac{n\sqrt{\lambda}}{\sqrt{2}} \right) = \sqrt{2}g(\lambda).
	\end{align*}
	So the result is proven.
\end{proof}

We can now bound $\lp$ and $\opt$ above and below in terms of the square root function.
\begin{lemma}
	Let $q(\lambda) = \frac{3n}{4}\sqrt{2\lambda}.$ For all $\lambda \in \left[\frac{8}{n^2},\frac12\right]$, 
	\begin{equation}
	\label{eq:qbound}
	q(\lambda) \leq \lp(\lambda) \leq \opt(\lambda) \leq \frac{4\sqrt{2}}{3} q(\lambda).
	\end{equation}
\end{lemma}
\begin{proof}
	First we show the relationship between $g$ and $q$:
	\begin{align*}
	q(\lambda) &= \frac{3n}{4}\sqrt{2\lambda}= n \left(\sqrt{2\lambda} - \frac{\sqrt{2\lambda}}{4} \right)\leq n \left( \sqrt{2\lambda} - \frac{2\lambda}{4}\right)= g(\lambda) \leq n\sqrt{2\lambda} = \frac{4}{3}q(\lambda).
	\end{align*}
Together with the bounds from Lemma~\ref{lemma:gbounds}, we get the final result:
\begin{equation}
\label{eq:finalbounds}
q(\lambda) \leq g(\lambda) \leq \lp(\lambda) \leq \opt(\lambda) \leq \sqrt{2}g(\lambda) \leq  \frac{4\sqrt{2}}{3} q(\lambda). 
\end{equation}
\end{proof}

\subsection{Establishing the $\Omega(\log n)$ Lower Bound}
At this point we have shown that both the LambdaPrime ILP as well as its LP relaxation are bounded above and below in terms of the square root function. We recall now a result of Magnanti and Stratila~\cite{magnanti2012separable} that lower bounds the number of linear pieces needed to approximate the square root function in a given interval. Afterwards, we will use this to show that the inequality in~\eqref{eq:finalbounds} implies a similar lower bound on the number of feasible solutions needed to approximate the LambdaPrime objective or its LP relaxation in every parameter regime on the ring graph.
\begin{theorem}
	\label{thm:magnanti}
	(Theorem 2 in~\cite{magnanti2012separable}, adapted.)
	Let $\rho: \mathbb{R}_+ \rightarrow \mathbb{R}_+ $ be a piecewise-linear function that approximates $\psi(x) = \sqrt{x}$ to within a factor $q > 1$ on $[a,b]$. Then $\rho$ must contain at least $\ceil*{\frac{1}{3} \log_{\gamma(q)} \frac{b}{a}}$ linear pieces, where $\gamma(q) = (2q^2 - 1 + 2q\sqrt{q^2 - 1})^2$.
\end{theorem}
A proof of the result is given in the original work~\cite{magnanti2012separable} for the case where $q = (1+\varepsilon)$ for some~$\varepsilon > 0$. We have slightly adapted the statement to use a general approximation factor~$q > 1$, and have slightly changed notation. We note that the result does not change if we scale the function $\psi(x) = \sqrt{x}$ by a constant.

We use this result in order to establish the need for~$\Omega(\log n)$ feasible solutions to approximate the LambdaPrime LP relaxation on ring graphs. To do so we start with a general lemma regarding approximation via piecewise-linear functions. To state the lemma, we will say that a function~$\psi$ is a $p$-approximation for another function~$\phi$ for all~$x \in [a,b]$, if for every~$x \in [a,b]$ we have $\phi(x) \leq \psi(x) \leq p \phi(x)$. 
\begin{lemma}
	\label{lemma:hardtransfer}
	Let $M > 1$ be a constant and let $\phi$ and $\psi$ be functions satisfying:
	\[ 0 \leq \phi(x) \leq \psi(x) \leq M \cdot \phi(x) \text{ for all $x \in [a,b]$ where $0 < a < b$.}\]
	\begin{enumerate}
		\item If $\rho: [a,b] \rightarrow \mathbb{R}_{+}$ is a piecewise-linear function such that for some $p > 1$, $\rho$ is a $p$-approximation for $\psi$ for all $x \in [a,b]$, then $\rho$ is a $(pM)$-approximation for $\phi$ for all $x \in [a,b]$.
		\item When approximating $\phi$ by a piecewise-linear function, if it takes at least $B$ linear pieces to obtain a $(pM)$-approximation for $\phi$, then it takes at least $B$ linear pieces to get a piecewise-linear $p$-approximation for $\psi$.
	\end{enumerate}
\end{lemma}
\begin{proof}
	The first statement is simply the observation that for every $x \in [a,b]$,
	\[\phi(x) \leq \psi(x) \leq \rho(x) \leq p \,\psi(x) \leq (p M)\, \phi(x).\]
	
	The second statement follows by contradiction. Assume that we have obtained a piecewise linear function $\rho$ with fewer than $B$ linear pieces that is within a factor $p$ of $\psi$ for every $x \in [a,b]$. By the first statement, this must mean that $\rho$ is a $(pM)$-approximation for $\phi$ with fewer than $B$ linear pieces, but this is a contradiction.
%
\end{proof}
Finally, we prove our main result of the section.
\begin{theorem}
	\label{thm:logn}
	Let $\textbf{X}$ denote a set of feasible solutions to the LambdaPrime LP on the ring graph $G_k$ for any $k \geq 3$. Assume that for every $\lambda \in \Big[\frac{8}{n^2}, \frac{1}{2}\Big]$, $\textbf{X}$ contains a feasible LP solution $(x_{ij})$ that is within a factor $p > 1$ of the optimal LP relaxation score. Then $\textbf{X}$ must contain at least $B = \ceil*{\frac{2}{3} \log_{\gamma(pM)} \frac{n}{4}}  = \Omega(\log n)$ feasible LP solutions, where $M = 4\sqrt{2}/{3}$ and $\gamma(x) = (2x^2 - 1 + 2x\sqrt{x^2 - 1})^2$.
\end{theorem}
\begin{proof}
	This result is simply a combination of Lemma~\ref{lemma:hardtransfer}, Theorem~\ref{thm:magnanti}, and the bounds:
	\begin{equation}
	q(\lambda) = \frac{3n}{4}\sqrt{2\lambda} \leq \lp(\lambda) \leq \frac{4\sqrt{2}}{3} q(\lambda). 
	\end{equation}
	We are considering the interval $[a,b]$ where $a = \frac{8}{n^2}$ and $b = \frac{1}{2}$. By Theorem~\ref{thm:magnanti}, in order to get a piecewise-linear $(pM)$-approximation for $q$, it takes at least $B = \ceil*{\frac{1}{3} \log_{\gamma(pM)} \frac{b}{a}} = \ceil*{\frac{2}{3} \log_{\gamma(pM)} \frac{n}{4}} $ linear pieces. By Lemma~\ref{lemma:hardtransfer} then, every piecewise-linear $p$-approximation for $\lp$ must also have $B$ linear pieces.
	
	Recall now that for any feasible LP solution $(x_{ij}^k) \in \textbf{X}$, we can compute $P_k = \sum_{(i,j) \in E} x_{ij}$ and $N_k = \sum_{i<j} (1-x_{ij})$, so that as $\lambda$ varies, the LP score corresponding to $(x_{ij}^k)$ is a linear function $P_k + \lambda N_k$. Thus, the {family} of LP solutions $\textbf{X}$ defines an entire set of linear functions. For each $\lambda \in [8/n^2,1/2]$, we can select the LP solution which minimizes $P_k + \lambda N_k$ over $k \in\{1, 2, \hdots, |\textbf{X}| \}$, in order to build new a piecewise-linear function from $\textbf{X}$ that \emph{approximates} the piecewise-linear function $\lp (\lambda)$ for the given $\lambda$ values. If $\textbf{X}$ contained fewer than $B$ feasible LP solutions, we could use it to construct a piecewise-linear function with fewer than $B$ pieces that approximates $\lp$ to within a factor $p > 1$ for all $\lambda \in \big[\frac{8}{n^2}, \frac{1}{2} \big]$. This would contradict the conclusions of Theorem~\ref{thm:magnanti} and Lemma~\ref{lemma:hardtransfer}.
\end{proof}
Observe that in the proof of Theorem~\ref{thm:logn}, all steps will hold in exactly the same way if we consider the LambdaPrime ILP rather than the LP, so this lower bound also holds for finding approximately optimal clusterings for the LambdaPrime objective.

\subsection{Special Case: Star Graphs}
Although in the worst case we need $\Omega(\log n)$ feasible solutions to approximate the LambdaPrime LP in all parameter regimes, there are cases where this worst-case scenario does not hold. We illustrate this point by proving that for star graphs, a single LP solution is optimal for all (nontrivial) values of $\lambda$.

Recall from Section~\ref{sec:star} that for a star graph on $n$ nodes, one needs exactly $n-1$ clusterings in order to solve the LambdaPrime objective in all parameter regimes. For this graph, the minimum scaled sparsest cut is $\lambda_1 = \frac{1}{n-1}$, and for any $\lambda \geq \frac{1}{2}$, the optimal solution is to place one node with the center node, and put each other node in a singleton cluster. Despite the need for a linear number of clusterings to fully capture the optimal solution space (which we showed in Section~\ref{sec:star}), we have the following result about the LP relaxation.
\begin{theorem}\label{thm:starlp}
	Let $G$ be an $n$-node star graph where node 1 is the center node. Define a feasible solution by setting $x_{1i} = \frac{1}{2}$ for all $i = 2, 3, \hdots, n$, and by setting $x_{ij} = 1$ for $i \neq j$ when $i\neq 1$ and $j \neq 1$. This feasible solutions will optimally solve the LambdaPrime LP relaxation for every $\lambda \in \big(\frac{1}{n-1}, \frac{1}{2} \big)$.
\end{theorem}
\begin{proof}
	Following Step 1 of the proof of Theorem~\ref{thm:ringlp}, we can show that for $\lambda \in (1/(n-2), 1/2)$, there exists an optimal LP solution on the star graph in which all positive edge distances~$x_{1i}$ are the same. In more detail, given any arbitrary LP solution $\textbf{x}^1 = (x_{ij}^1)$, we can use the symmetry of the star graph to construct $n-2$ other optimal LP solutions by mapping outer nodes to each other. For $t = 1, 2, \hdots, (n-2)$, let $\textbf{x}_t$ be the LP solution obtained by mapping node $i \neq 1$ to node $i+t$. Taking a convex combination of all the resulting LP solutions produces a feasible solution in which $x_{1i} = c$  for all $i = 2, 3, \hdots , n$, for some $c \geq 0$. Furthermore, any two non-central nodes $i$ and $j$ are separated by 
	a path of two positive edges, so $x_{ij} = \min \{1, 2c \}$ for all $i \neq j$ such that $i \neq 1$, $j\neq 1$.
	
	If we consider only LP solutions of this type, for any $c \leq \frac{1}{2}$, the LambdaPrime LP relaxation on the star graph would be
	\begin{equation}
	\label{eq:starlp}
	\lp(\lambda) = (n-1)c + \lambda {n \choose 2}(1 - 2c) = (n-1) \left( c (1- \lambda n) + \frac{\lambda n}{2} \right),
	\end{equation}
	and since $\lambda > 1/(n-1) > 1/n \implies (1-\lambda n) < 0$, the above is minimized if $c =1/2$. If instead we restrict to $c \geq 1/2$, then the LP relaxation would be $\lp(\lambda) = (n-1)c$, which is clearly minimized if $c = 1/2$. Thus, for star graphs, a single LP solve is sufficient to approximate the LambdaPrime LP relaxation in all parameter regimes.
\end{proof}

\section{Frontier Extension Algorithm}
\label{sec:fe}
We do not \emph{always} need~$\Omega(\log n)$ feasible solutions to approximate the \lpcc LP relaxation in all parameter regimes.
So we now turn our attention to more sophisticated algorithms for this approximation task.
Given a feasible solution which optimizes the LP relaxation at given value of~$\lambda$, we develop a technique for explicitly computing the exact range of~$\lambda$ values for which this solution is optimal or nearly optimal.
This relies on adapting existing work in sensitivity analysis for parametric linear
programming~\cite{jansen1997sensitivity,nowozin2009solution}. With this technique as a primitive, we develop an approach
for more carefully selecting values~$\lambda$ at which to evaluate the LP relaxation, in order to eventually obtain a
small family of clusterings which approximate \lpcc in all regimes. We first introduce new terminology that will be helpful in the developement of our new methods. 
\begin{definition}
For a fixed $\lambda_0 \in (0,1)$, let~$\textbf{x}_0$ denote the optimal solution to the \lpcc LP relaxation. We
define the \emph{optimal range} of~$\textbf{x}_0$ to be the interval of~$\lmd$ values for which~$\textbf{x}_0$ optimizes
the LP relaxation, and denote this interval typically by~$[\alpha,\beta]$.
\end{definition}
Existing work on parametric linear programming confirms that every feasible solution that is optimal for a single parameter is in fact optimal over an entire range of parameters~\cite{jansen1997sensitivity}.
The concept of the optimal range can be naturally extended to approximation of the optimal objective.
\begin{definition}
For a fixed $\lambda_0 \in (0,1)$, let~$\textbf{x}_0$ be the optimal solution to the \lpcc LP relaxation. Define
the~\emph{$\eps$-approximate range} as the interval of~$\lmd$ values such that~$\textbf{x}_0$ obtains
a~$(1+\eps)$-approximation,
and denote this interval typically by~$[a,b]$.
\end{definition}
When the union of~$\eps$-approximate ranges of several solutions contains an interval, we say that these solutions \emph{cover} that interval.
If the union of the~$\eps$-approximate ranges of several LP solutions contains the~$\eps$-approximate range for another solution~$\mbx^*$, we say that these solutions cover~$\mbx^*$.

\subsection{Obtaining Approximate and Optimal Ranges}
\label{orlp}
Nowozin and Jegelka~\cite{nowozin2009solution} provide a framework for computing what they define as the \emph{stability
range}, for a variety of clustering problems. Given a solution to a clustering objective, these authors consider a
perturbation vector $\textbf{d}$ and a parameter $\theta$. They then compute the range of $\theta$ values for which the
given clustering remains optimal even when the objective function is perturbed by $\theta \textbf{d}$. Adding or
subtracting such a vector $\theta \textbf{d}$ from the objective is equivalent to perturbing the weights in a clustering
problem by a small amount. Thus, a clustering which remains optimal for a wide range of~$\theta$ values represents a particularly stable clustering. 

Our definition of an \emph{optimal range} can be viewed as a special case of Nowizin and Jegelka's stability range. In
our case, we do not perturb a clustering objective by an arbitrary noise vector or perturbation vector; we specifically
perturb the objective by changing the resolution parameter $\lambda$ slightly. Since this is just a special case of
perturbing the weights of \lpcc, we can adapt the work of Nowizin and Jegelka~\cite{nowozin2009solution} to compute the
optimal range, as well~$\eps$-approximate range, for a solution to the \lpcc LP relaxation.

\paragraph*{Formulation}
In light of this framework~\cite{nowozin2009solution}, we formulate  \lcclp in the following
matrix form, as \textbf{P1}:
\begin{equation}
\begin{array}{lll}
\text{minimize }  & \text{\textbf{P1}} =  & \mbc^T \mbx +  \lmd
\binom{n}{2}\\
\subjectto & & A\mbx \ge \mbb\\
& & \mbx \ge 0\,.
\end{array}
\end{equation}
Here matrix~$A$ encompasses both the triangle inequalities ($x_{ij} \le x_{ik}+x_{kj}$, for all $i,j,k$) and also the upper bounds of the $x_{ij}$ variables ($x_{ij} \le 1$ for all $i<j$).
Vector~$\mbc$ encodes the \lpcc weights, shifted somewhat:~$c_{ij}
= 1- \lmd$ if~$\{i,j\} \in E$; and~$c_{ij} = -\lmd$ if~$\{i,j\}\not\in E$.

For a fixed resolution parameter $\lambda_0$, let~$\mathbf{x}^* = (x^*_{ij})$ denote the optimal \lpcc solution.
Because the coefficient of~$\lambda$ in the term~$ \mbc^T \mbx$ (of the \lpcc objective) is~$-1$,
when we shift~$\lambda$, we arrive at the perturbed vector~$\mbc' = \mbc +s\theta \mbd$, where~$\mbd$ has $d_{ij}=-1$
for all~$i<j$. Here, $s$ is either~$+1$ for forward perturbation or~$-1$ for backward perturbation,
and~$\theta$ represents the \emph{extent} of perturbation.
Let~$\lmd'$ be~$\lmd_0+s\theta$.
Then~$(\mbc')^T \mbx + \lmd'\binom{n}{2}$ is the \lcclp objective for~$\lmd'$. 

Based on the dual of \textbf{P1},
we define a new problem \orlp (\textbf{O}ptimal \textbf{R}ange \textbf{LP})
With input parameters $\mbx^*, s=\pm1, \lmd$ and~$\eps$ are input parameters, it
seeks the maximum
perturbation~($\theta$) while maintaining the approximation factor of the \lcclp solution~$\mbx^*$.
In the perturbed setup, the dual constraint is 
$A^T\mby \le \mbc+s\theta \mbd$.
Should we want to ensure the optimal dual objective is the same as the primal objective for solution~$\mbx^*$, we would arrive at the constraint,
\begin{equation}\label{eq:duality_con}
\mbb^T\mby + \lmd' \binom{n}{2} = (\mbc')^T\mbx^* + \lmd' \binom{n}{2}\,.
\end{equation}
The left hand side of Constraint~\eqref{eq:duality_con} is the \lcclp objective for~$\lmd'$, and the equality only holds when~$\mbx^*$ is  optimal for~$\lmd'$.
For all~$\mby \ge 0$ satisfying the dual constraint,
\[
\max_{\mby} \mbb^T \mby = \mathbf{LP}(\lmd') \le (\mbc')^T\mbx^* + \lmd' \binom{n}{2}\,.
\]
To show that $\mbx^*$ remains a~$(1+\eps)$-approximate solution to~$\lcclpNone(\lmd')$, we find \emph{some}
dual solution~$\mby$ whose objective is within factor~$1+\eps$ of the primal objective for~$\mbx^*$, but for
the setting of~$\lmd'$.
That is,
\[
(1+\eps)\left(\mbb^T\mby + \lmd' \binom{n}{2}\right) \ge (\mbc')^T\mbx^* + \lmd' \binom{n}{2}\,.
\] 
With these constraints, our aim is to maximize the interval over which~$\mbx^*$ remains a near-optimal solution, viz.
\begin{align*}
&\text{maximize} \quad \text{\orlp} =\theta  \\
&\subjectto \quad (1+\eps)\left(\mbb^T\mby + (\lmd_0+s\theta) \binom{n}{2}\right) \ge \mbc^T\mbx^*+s\theta \mbd^T\mbx^* + (\lmd_0+s\theta)\binom{n}{2} \\
&\hspace{20mm} A^T\mby \le \mbc+s\theta \mbd \\
&\hspace{20mm} \mby \ge 0 \\
&\hspace{20mm} \theta \ge 0\,.
\end{align*}
Let~$\theta^+$ be the optimal objective value of \textbf{ORLP} with~$s=1$, and let~$\theta^-$ be the optimal objective value with~$s=-1$.
Since \lcclp as a function of~$\lmd$ is piecewise linear and concave, solution~$\mbx^*$ is a~$(1+\eps)$-approximation to \lcclpNone$(\lmd)$ for all~$\lmd \in [\lmd_0 - \theta^-, \lmd_0 + \theta^+]$. 

\begin{theorem}\label{thm:optimal_range}
(adapted from Section 2.1 of Nowozin and Jegelka, 2009\cite{nowozin2009solution})
Given an instance of \lpcc, let~$\mbx^*$ be the optimal LP solution of~$\lcclpNone(\lmd_0)$.
For every~$\eps\ge0$, solution~$\mbx^*$ is a~$(1+\eps)$-approximation exactly in the range~$\lmd \in [\lmd_0 - \theta^-,
\lmd_0+\theta^+]$, where~$\theta^+$ and~$\theta^-$ are optimal objective values of~\lpPlus and~\lpMinus, respectively.
\end{theorem}
\subsection{The Algorithm}

In this section we discuss the Frontier Extension (FE) algorithm for finding an approximation to the LambdaPrime LP.
It invokes \orlp as a primitive to return a family of optimal solutions such that for every~$\lmd \in (4/n^2,1)$ (the
range is justified in Section~\ref{sec:log_lp_eval}), there is
a solution in the family which is a~$(1+\eps)$-approximation to the \lpcc LP relaxation.
%
We define the following concepts for the optimal \lcclp solutions.
\begin{definition}
A set of optimal solutions~$\curC$ of the LP relaxation is called an~\emph{$\eps$-cover}, if for all~$\lmd_0 \in (4/n^2, 1)$, there exists a solution~$\mbx \in \curC$ such that $\mbx$ is a~$(1+\eps)$-approximation for \lcclpZero.
\end{definition}
The goal of FE algorithm is to find an~$\eps$-cover for~$\lmd \in (4/n^2, 1)$. An~$\eps$-cover is~\emph{optimal} if it has a minimum number of LP solutions among all~$\eps$-covers. 
When context is clear we drop the~$\eps$ and call it simply an \emph{optimal cover}.


As a convention, we denote the optimal range of an LP solution by~$[\alpha, \beta]$ and the~$\eps$-approximate range by~$[a,b]$.
The FE algorithm, laid out as Algorithm~\ref{alg:fe}, provides effective way to approximate the LambdaPrime LP in all parameter regimes by alternating between solving the LP relaxation, and then finding the $\eps$-approximate range of the last computed solution to the original LP. In this way, it greedily pushes the frontier of the $\lambda$ values that are guaranteed to be ``$\eps$-covered'' by some previously computed solution to the LambdaPrime LP relaxation.
\begin{algorithm}[ht]
\caption{Frontier Extension (FE)}\label{alg:fe}
\begin{algorithmic}[1]
\Function{FE}{$G, \eps$}
\State $\curC \gets \varnothing$
\State $\lmd_0 \gets 4/n^2$
\While{$\lmd_0 < 1$}
\State $\mbx^* \gets$ \lcclpZero on~$G$
\State $\theta_+ \gets $ \lpPlus on~$G$
\State $\lmd_+ \gets \lmd_0 + \theta_+$ \label{ln:lmd-plus}
\State Add~$\mbx^*$ to $\curC$
\State $\lmd_0 \gets (1+\eps)\lmd_+ \label{ln:next_lmd}$
\EndWhile
\If{$\lambda_+ < 1$}
\State $\mbx^* \gets$ $\textbf{LP}(\lambda_+)$ on~$G$
\State Add~$\mbx^*$ to $\curC$
\EndIf
\State\Return $\curC$
\EndFunction
\end{algorithmic}
\end{algorithm}
Here \lcclpZero solves the \lpcc LP on input graph~$G$ with a parameter~$\lmd_0$ and \lpPlus solves the LP referred to
in Theorem~\ref{thm:optimal_range}.

For an optimal solution~$\mbx^i$, denote the optimal range of solution~$\mbx^i$ by~$[\alpha_i, \beta_i]$ and the~$\eps$-approximate range by~$[a_i, b_i]$:  the values~$a_i$ and~$b_i$ depend on~$\eps$.
For the optimal solutions of the \lpcc LP, we can define a natural total ordering.
For two solutions~$\mbx$ and~$\mbx'$, $\mbx \succ \mbx'$ if and only if~$\alpha>\alpha'$, where~$\alpha$ is the smaller endpoint of the optimal range of~$\mbx$, and~$\alpha'$ is that of~$\mbx'$.
\begin{theorem}
The FE algorithm produces an~$\eps$-cover.
\end{theorem}
\begin{proof}
    Let~$\lmd_1$ and~$\lmd_2$ be two consecutive~$\lmd$ values computed by FE in step 3 or step 6.
    Let~$\mbx^1$ and~$\mbx^2$ denote the two optimal solutions from solving \lcclp for~$\lmd_1$ and~$\lmd_2$, respectively.
    Let~$\lmd_+$ be the value calculated at line~\ref{ln:lmd-plus} for~$\lmd_1$.
    For~$\lmd \in [\lmd_+, \lmd_+(1+\eps)]$, these~$\lmd$ values are not covered by~$\mbx^1$, but by~$\mbx^2$.
  Lemma~\ref{lem:approx} implies that~$\mbx^2$ is a~$(1+\eps)$-approximation for~$[\lmd_2 / (1+\eps),\lmd_2] = [\lmd_+, \lmd_+ (1+\eps)]$. Thus, the interval between $\lmd_+$ and $\lmd_+ (1+\eps)$ will be covered by the LP solution computed in the next iterate. Note though that after the last step of the while loop, if $\lambda_+ (1+\eps) > 1$ but $\lambda_+ < 1$, the loop will terminate before we have computed an LP solution covering the interval $(\lambda_+, 1)$. Thus, we compute the LP relaxation at the last $\lambda_+$ value in order to ensure we have a complete cover.
\end{proof}
We analyze the performance of the FE algorithm to provide ways to bound the number of LPs we need to solve, and the size of the solution set discovered by FE.
In the following text, we call the family of solutions discovered by the algorithm the~\emph{FE cover}.
\subsection{Parameterized Analysis}
Our motivation for bounding the number of LPs we need to solve is two-fold.
We first show that for any optimal cover (i.e., a cover with a minimal number of LP solutions), we can replace every solution in that optimal cover by two or fewer solutions in the FE cover, and still maintain the~$\eps$-approximation in all parameter regimes. This allows us to refine the FE cover to produce a new cover that is within a factor two of the optimal cover in terms of size.
Our second reason for bounding the number of LPs we need to solve is simply to guarantee that we will not have to solve the LP relaxation too many times in order to produce our final~$\eps$-cover.

\begin{lemma}\label{le:it_takes_two}
    Let~$\curC'$ be an optimal cover, and~$\curC$ be the FE cover.
    We index solutions in~$\curC$ by the natural ordering based on the left end point of the optimal range.
    For any~$\mbx_{*}^i \in \curC'$, if $\mbx_{*}^i \not\in \curC$, it can be covered by the union of two solutions in~$\curC$.
    Specifically, the latter case means that there exist two consecutive solutions~$\mbx^{j}$ and~$\mbx^{j+1}$
in~$\curC$ such that the union of their $\eps$-approximate ranges contain the $\eps$-approximate range of~$\mbx_{*}^i$.
\end{lemma}
\begin{proof}
    Assume that~$\mbx_{*}^i \not\in \curC$. Note that there must exist a pair of two consecutive solutions~$\mbx^{j}$ and~$\mbx^{j+1}$ in~$\curC$ such that $\mbx^{j} \prec \mbx_{*}^i \prec \mbx^{j+1}$.
    

    During the process of the FE algorithm,~$\mbx^{j+1}$ is obtained by solving the \lpcc LP with~$\lmd=(1+\eps)b_j$, as~$\lmd=b_j$ is the last point on which~$\mbx^{j}$ is a~$(1+\eps)$-approximation.
	If~$(1+\eps)b_j$ falls into~$[\alpha^*_i, \beta^*_i]$, the optimal range of~$\mbx_{*}^i$, then FE would add~$\mbx_{*}^i$ to the solution set~$\curC$.
	Since we assumed~$\mbx_{*}^i$ is not in the FE cover, we must have~$(1+\eps)b_j > \beta^*_i$. Due to the concavity of the optimal objective function, $\mbx^{j}$ is a better approximation than~$\mbx_{*}^i$ for every~$\lmd < a_i$. Therefore, $\mbx^{j}$ is a $(1+\eps)$-approximation from $a_j$ to $b_j$, where $a_j \leq a_i^*$. Similarly, $\mbx^{j+1}$ is a~$(1+\eps)$-approximation from $a_{j+1} \leq b_j$ to $b_{j+1}$, and due to concavity, it is a better approximation than~$\mbx_{*}^i$ for every~$\lambda > b_{j+1}$. Therefore, $b^*_i \leq b_{j+1}$.
    In conclusion, the union of the $\eps$-approximate ranges of~$\mbx^{j}$ and~$\mbx^{j+1}$ contains that of~$\mbx_{*}^i$.
\end{proof}
Note that although~$\curC$ only needs to use two solutions to cover a solution from~$\curC'$ that is not in~$\curC$, it does not follow that~$\curC$ can find the two solutions without redundancy. 
For instance, it might take more than four solutions to cover the same range of two consecutive solutions~$\mbx_{*}^i$ and~$\mbx_{*}^{i+1}$ in~$\curC'$, because it might take multiple iterations for FE to fill up the gap between the pair of solutions covering~$\mbx_{*}^i$ to the pair covering~$\mbx_{*}^{i+1}$.
We address this issue by a parameterized analysis, relying on the following new definitions.
\begin{definition}\label{def:forward}
    Given an instance of~\lpcc, for a fixed~$\eps$ and a family of solutions~$\curC$, each of which is optimal for some
$\lambda$, index the solutions in~$\curC$ in increasing order and denote the~$i^{\text{th}}$ solution by~$\mbx^i$.
    
    The \emph{forward ratio~$w(\mbx^i)$ of solution~$\mbx^i$ in~$\curC$} is defined as
    \[
    w(\mbx^i) =\begin{cases}
    \frac{1}{\beta_i} & \text{if } i = \abs{\curC} \\
    \frac{\alpha_{i+1}}{\beta_i} & \text{otherwise.}
    \end{cases}
    \]
    The \emph{forward factor~$p(\eps, \curC)$} is defined as $p(\eps, \curC) = \max\left\{ \max_i \{\ceil{\log_{1+\eps} w(\mbx^i)} \} ,0 \right\}$.
\end{definition}
Note that~$\curC$ does not necessarily have to be an~$\eps$-cover: it is simply an arbitrary set of optimal solutions.
The forward ratio~$w(\mbx^i)$ measures the gap between the right endpoint of the optimal range of~$\mbx^i$ and left endpoint of the optimal range of~$\mbx^{i+1}$, for~$\mbx^i, \mbx^{i+1} \in \curC$.
We only consider values of~$\ceil{\log_{1+\eps} w(\mbx^i)}$ that are non-negative in Definition~\ref{def:forward};
otherwise the gap between~$\mbx^i$ and~$\mbx^{i+1}$ could be $\eps$-covered by the intervals around~$\mbx^i$
and~$\mbx^{i+1}$ already, and no iteration would be needed to close the gap.
    
Moreover, we can bound the growth rate of the frontier by the following lemma.
\begin{lemma}\label{le:grow_rate_r}
    Let~$\curC$ be the FE cover.
    The number of solutions in~$\curC$ that are calculated to cover interval~$[\lmd_1, \lmd_2]$ is at most~$\ceil{(1/2)\log_{1+\eps} (\lmd_2/\lmd_1)}$.
\end{lemma}
\begin{proof}
    From Lemma \ref{lem:approx} and the definition of FE, we know that at each iteration, the right most point of the coverage is pushed forward by at least a multiplicative factor~$(1+\eps)^2$.
    The number of iterations it takes from~$\lmd_1$ to~$\lmd_2$ is then bounded by
    \[
   \left\lceil  \log_{(1+\eps)^2} (\lmd_2/\lmd_1) \right\rceil =  \left\lceil (1/2)\log_{1+\eps} (\lmd_2/\lmd_1) \right\rceil \,.
    \]
\end{proof}
\begin{lemma}\label{le:general_bound}
Let~$\curC'$ be a family of LP solutions, each of which is optimal for some $\lambda$. Let $p' = p(\eps, \curC')$, and let~$\curC$ be the FE cover.
We have~$\abs{\curC} \le (p'/2+1)\abs{\curC'}$.
\end{lemma}
\begin{proof}

Index solutions in~$\curC'$ by the natural ordering,~$\prec$.
Consider two consecutive solutions,~$\mbx_{*}^i$ and~$\mbx^{i+1}_*$, in~$\curC'$.
By Lemma~$\ref{le:it_takes_two}$, there exists a pair of solutions, $\mbx^{j}$ and~$\mbx^{j+1}$ in~$\curC$, whose union covers~$\mbx_{*}^i$, and a pair of solutions, $\mbx^{j'}$ and~$\mbx^{j'+1}$ whose union covers~$\mbx^{i+1}_*$.
For the FE algorithm to reach from~$\mbx^{j+1}$, the second solution of the first pair, to~$\mbx^{j'}$, the first solution of the second pair, it might take multiple iterations.
We then bound the number of iterations in the rest of the proof.
Note that FE needs to cover~$\lmd$ values between~$\lmd= b_{j+1}$ and somewhere to the left of~$\lmd = \alpha^*_{i+1}$.
If for some~$i$ we have~$b_{j+1} > \alpha_{i+1}^*$, then it is the ideal case, as~$\mbx^{i+1}_*$ will either be in~$\curC$ or the~$\eps$-approximate range of~$\mbx_*^{i+1}$ is entirely covered by another solution in~$\curC$. 
Nonetheless, set~$\curC$ does not need more than one solution for the~$\lmd$ values in the~$\eps$-approximate range of~$\mbx^{i+1}_*$.

Consider the harder case where~$b_{j+1} < \alpha^*_{i+1}$.
From Lemma \ref{le:grow_rate_r} we know the amount iterations it takes from~$b_{j+1}$ to~$\alpha^*_{i+1}$ is bounded by~$\ceil{(1/2)\log_{1+\eps} (\alpha^*_{i+1}/b_{j+1})}$.
We have
\[
\left \lceil \frac{1}{2}\log_{1+\eps} \frac{\alpha^*_{i+1}}{b_{j+1}}\right \rceil \le 
\left \lceil \frac{1}{2} \log_{1+\eps} \frac{\alpha^*_{i+1}}{b^*_{i}} \right  \rceil  =
\left \lceil \frac{1}{2} \log_{1+\eps} \left(\frac{\beta^*_i}{b^*_i} \frac{\alpha^*_{i+1}}{\beta^*_{i}} \right) \right\rceil  \le 
\left\lceil \frac{p'-1}{2} \right\rceil
\le \frac{p'}{2}.
\]
The second-last inequality arises from~$b^*_i / \beta^*_i \ge 1+\eps$.
Hence there can be at most~$p'/2$ solutions in~$\curC$ between~$\mbx_{*}^i$ and~$\mbx^{i+1}_*$.
To cover~$\mbx^{i+1}_*$, only one more solution is required.
Therefore~$\abs{\curC} \le 1+\sum_{i=1}^{\abs{\curC'}-1} (p'/2+1) \le (p'/2+1)\abs{\curC'} $.
If~$p'=0$ then it is easy to show that~$\abs{\curC} \le \abs{\curC'}$.
\end{proof}
Lemma \ref{le:it_takes_two} and Lemma~\ref{le:general_bound} together demonstrate the strength of FE algorithm when we
have some knowledge of a set of optimal solutions. Note finally that we can provide a simple upper bound on the number of LP solutions output by the FE algorithm.
\begin{theorem}\label{thm:upper_bound_fe_cover}
	Let~$\curC$ be the FE cover. We have $\abs{\curC} \le  \ceil{\log_{1+\eps}(n)}$.
\end{theorem}
\begin{proof}
	This bound is obtained simply by applying Lemma~\ref{le:grow_rate_r} with $\lambda_1 = 4/n^2$ and $\lambda_2 = 1$:
	\[
	\left\lceil \frac12 \log_{1+\eps}\left(\frac{1}{4/n^2} \right) \right\rceil \le \left\lceil \log_{1+\eps}(n) \right\rceil  .
	\]
\end{proof}



We end this subsection by considering a lower bound on the number of LP solutions needed to optimize the LambdaPrime LP relaxation in all parameter regimes. 
We can bound the optimal cover size via the forward factor directly.
\begin{theorem}\label{thm:lower_bound_opt_cover}
For a given \lpcc instance, let~$\curC^*$ be the optimal cover and~$\curC$ be an FE cover.
Let~$p^* = p(\eps,\curC^*)$.
We have~$\abs{\curC} \le (p^*/2+1)\abs{\curC^*}$.
\end{theorem}
Theorem \ref{thm:lower_bound_opt_cover} is proved by taking~$\curC^*$ as the set of solutions in
Lemma~\ref{le:general_bound}. It establishes a performance bound for the FE algorithm, by characterizing the optimal cover, and also serves as a lower bound for the optimal cover.

Our results here apply to general instances of LambdaPrime, in contrast to Theorem \ref{thm:logn}, where the bound is proven specifically for ring graphs. Our main result of the next subsection is to show how we can refine the output of the FE algorithm to produce a family of LP solutions that has at most twice the number of LP solutions the optimal family. 
\subsection{Refining the FE Output}
\begin{algorithm}[t]
	\caption{Frontier Extension, Backward Elimination (FEBE)}\label{alg:febe}
	\begin{algorithmic}[1]
		\Function{FEBE}{$G, \eps$}
		\State $\curC' \gets \varnothing$
		\State $\curC \gets$ FE$(G,\eps)$
		\State For each $\textbf{x}_i \in \mathcal{C}$, compute $\textbf{ORLP}(\textbf{x}_i,-1,\lambda_i, \eps)$ to obtain $\eps$-approximate range $[a_i, b_i]$ 
		\State $\lambda_0 = \frac{4}{n^2}$
		\While{$\lambda_0 < 1$}
		\State $S = \{\textbf{x}_i \in \mathscr{C} : \lambda_0 \in [a_i, b_i] \}.$
		\State $i^* = \argmax_{i: \textbf{x}_i \in S} b_i$.
		\State Add $\textbf{x}_{i^*}$ to $\curC'$.
		\State $\lambda_0 \leftarrow b_{i^*}$
		\EndWhile
		\State\Return $\curC'$
		\EndFunction
	\end{algorithmic}
\end{algorithm}
Although the cover returned by the FE algorithm can potentially be much larger than the optimal cover, the FE cover can
be ``distilled'' to a more compact, $\eps$-cover without redundant clusterings. Fix ~$\eps > 0$ and run the FE algorithm
to obtain a family of LP solutions~$\mathscr{C}$ that is an~$\eps$-cover. Let~$\mathscr{C}^*$ be an
optimal~$\eps$-cover. We know from Lemma~\ref{le:it_takes_two} that there exists a subfamily of~$\mathscr{C}$ that
contains at most $2 |\mathscr{C}^*|$ LP solutions but remains an~$\eps$-cover. Let~$\mathscr{C}' \subseteq \mathscr{C}$
be a subfamily of $\mathscr{C}$ of minimum size while still being an~$\eps$-cover. If we can find a way to extract such a sub-family~$\mathscr{C}'$, it must satisfy $|\mathscr{C}'| \leq 2|\mathscr{C}^*|$. 

We extract such a minimum size $\eps$-cover subfamily using the Frontier Extension Backward Elimination (FEBE) algorithm,
laid out in Algorithm~\ref{alg:febe}. FEBE begins computing the full $\eps$-approximate range for every LP solution in
the cover $\mathscr{C}$. More precisely, for every $\textbf{x}_i \in \mathscr{C}$, let $\lambda_i$ be the value of
$\lambda$ for which $\textbf{x}_i$ is optimal, and compute $\textbf{ORLP}(\textbf{x}_i,-1,\lambda_i, \eps)$. In this
way, for each $\textbf{x}_i \in \mathscr{C}$, we obtain the full $\eps$-approximate range $[a_i, b_i]$ for
$\textbf{x}_i$. Recall that $b_i$ was computed in the original FE algorithm. Note that, compared with the FE algorithm,
the FEBE algorithm increases the number of LPs solved by only a factor of two. 

Once we have extracted the exact $\eps$-approximate range for each $\textbf{x} \in \mathscr{C}$, the task of extracting
a minimum size $\eps$-cover subfamily is completely equivalent to extracting a minimum size family of subintervals of
the form $[a_i, b_i]$, in order to cover the interval of nontrivial~$\lambda$ values, $(4/n^2, 1)$. The general problem of finding a minimum cardinality covering of a larger interval by smaller interval arises in other settings as well, and permits a simple greedy solution which is the basis for FEBE~\cite{StackoverflowInterval}. For completeness we include full details.

FEBE starts at the left endpoint $\lambda_0 = 4/n^2$ and considers every interval $[a_i, b_i]$ that covers $\lambda_0$.
Among these, FEBE selects the interval with a maximum value for $b_i$, i.e., it selects the LP solution
$\textbf{x}_{i^*}$, where $i^*$ is the index maximizing $b_i$ among intervals covering $\lambda_0$. Intuitively, in order to end up with a subfamily that is still an~$\eps$-cover, FEBE must select at least one interval that covers $\lambda_0 = 4/n^2$, and so it selects the interval that ``reaches'' farthest to the right. The goal then becomes finding a cover for the subinterval $(b_{i^*}, 1)$, so the same procedure is repeated on the subinterval by setting $\lambda_0 = b_{i^*}$. 
Clearly, this procedure produces an~$\eps$-cover; due to the greedy nature of the method, in fact it produces
a minimum-size~$\eps$-cover. If~$S_0$ is the set of intervals containing~$\lambda_0$ and FEBE did not select any intervals from $S_0$, then it would not produce an~$\eps$-cover. If it selects an interval from $S_0$ which does not maximize the size of the right endpoints, then it would need to find a cover for the interval $(b, 1)$, where $b < b_{i^*}$, which cannot possibly be covered with a smaller number of intervals than the amount for $(b_{i^*},1)$. We conclude with a theorem.
\begin{theorem}\label{thm:2approx}
    Let~$\curC$ be the FE cover, $\curC'$ be the FEBE~$\eps$-cover, and $\curC^*$ be the optimal~$\eps$-cover.
    Let~$p' =p(\eps, \curC')$. 
    We have~$\abs{\curC'} \le 2\abs{\curC^*}$, and~$\abs{\curC} \le 2(p'/2+1)\abs{\curC^*} = (p'+2)\abs{\curC^*}$.
\end{theorem}
The second bound in the theorem is shown by applying Lemma~\ref{le:general_bound} on the FE cover as compared to the FEBE cover. In this way, FEBE allows us not only to obtain a small~$\eps$-cover, but also provides a potentially stronger a posteriori bound on $|\mathscr{C}^*|$ by computing $p'$ for $\mathscr{C}'$. This in turn gives an a posteriori bound the number of LP relaxations solved by the original algorithm FE, in terms of the minimum number of LP solutions needed to approximate the LP relaxation in all parameter regimes. 

\section{Related Work}
We list relevant related work on Correlation Clustering, graph clustering with resolution parameters, and existing work on approximating other types of parametric objectives across a range of parameters.

\paragraph*{Correlation Clustering}
Correlation clustering is an objective for partitioning signed graphs in such a way that positive edges tend to be inside clusters, and negative edges tend to cross between clusters. The problem was first introduced by Bansal et al.~\cite{Bansal2004correlation}, with several later papers independently showing an $O(\log n)$-approximation for arbitrarily weighted signed graphs~\cite{CharikarGuruswamiWirth2005,demain2006ccgen,EmanuelFiat2003}. Recently~\cite{Veldt:2018:CCF:3178876.3186110}, some of the authors introduced a parameterized version of the objective called LambdaCC. An instance of LambdaCC is given by a complete signed graph in which all negative edges have weight $\lambda \in (0,1)$, and all positive edges have weight $(1-\lambda)$. In follow up work~\cite{gleich2018ccgen}, the same authors showed that the LambdaCC LP relaxation has an $O(\log n)$ integrality gap. They also later proved new results for how learn the best value of $\lambda$ to use in practice~\cite{Veldt2019learning}.

\paragraph*{Graph Clustering Objectives with Resolution Parameter}
Many other graph clustering objectives rely on tunable resolution parameters and are closely related to LambdaCC and LambdaPrime. We previously demonstrated that the Hamiltonian objective of Reichardt and Bornholdt~\cite{ReichardtBornholdt2006} is equivalent to LambdaCC for the appropriate parameter settings~\cite{Veldt:2018:CCF:3178876.3186110}. This Hamiltonian objective in turn is also known to be equivalent at optimality to a generalization of the modularity objective that includes a resolution parameter~\cite{Arenas2008analysis,newman2004modularity}. The stability objective of Delvenne et al.~\cite{Delvenne2010stabilitypnas} is another generalization of modularity that includes a resolution parameter $t$. Roughly speaking, stability measures the probability that a random walker starting inside a cluster will end up inside the same cluster after a random walk of length $t$. Although Delvenne at al.~\cite{Delvenne2010stabilitypnas} proved that stability is related to a linearized version of the Hamiltonian objective of Reichardt and Bornholdlt~\cite{ReichardtBornholdt2006}, these objectives are not identical and do not necessarily share the same set of optimal solutions. Other parametric clustering objectives include the constant Potts model of Traag et al.~\cite{traag2011narrow} and a multiresolution variant of the map equation introduced by Schaub et al.~\cite{schaub2012multiscale}.

Despite the large number of parametric graph clustering objectives~\cite{Arenas2008analysis,Delvenne2010stabilitypnas,ReichardtBornholdt2004,ReichardtBornholdt2006,traag2011narrow,schaub2012multiscale} and accompanying heuristic algorithms with tunable resolution parameters~\cite{genLouvain_software,traag2019leiden}, there has been little to no focus on finding optimal or near optimal families of clusterings for an objective across a range of parameter settings. This is at least in part due to the challenges associated with approximating the modularity objective~\cite{Dinh:2015:NCV:2919336.2920600}. Thus, our results provide a rigorous theoretical foundation for a problem that has been explored at length in the applied community detection literature.

\paragraph*{Approximation Algorithms for Parametric Objectives}
Our work also shares similarities with research on solving or approximating other types of objective functions for a range of input parameters. For instance, Mettu and Plaxton~\cite{Mettu:2003:OMP:639069.773507} introduced the \emph{online} $k$-median problem where locations appear and must be assigned sequentially \emph{and the value of $k$ is not known in advance}, and so the goal is that when $k$ is finally chosen, there is a solution that is not far from the optimal offline placement of $k$ locations. Incremental, unknown $k$, versions, of several other problems, including $k$-vertex cover, $k$-minimum spanning tree, and $k$-set cover, have also been considered in practice. Lin et al.~\cite{ling2010incremental} present an overview of previous work on incremental problems, as well as a general approach for incremental approximation algorithms. 

In our work we seek multiple solutions (i.e., a family of clusterings), which provide an approximate solution for an objective in all parameter regimes. A distinct though related goal is to find a single solution to an optimization problem that provides the best approximate possible for an entire range of parameters. This is the aim of related research on all-norm minimization~\cite{azar2002allnorm,golovin_et_al:LIPIcs:2008:1753}. Azar et al. introduced the concept of all-norm approximation algorithms, where the goal is to obtain a feasible solution to an optimization problem that such that the $p$-norm of a certain expression is small for all values of $p$. Later, Golovin et al.~\cite{golovin_et_al:LIPIcs:2008:1753} presented new results for all $p$-norm variants of set cover. Related problems in bandwidth allocation have also been considered by Kleinberg et al.~\cite{kleinberg2001} and Goel et al.~\cite{Goel:2001:AMF:365411.365483}.

\section{Discussion and Future Work}
Solving a parametric graph clustering objective for a range of resolution parameters is a useful way to detect clustering structures that highlight different characteristics in the same graph. This approach has been used frequently in applied research on community detection, with applications to hierarchical graph clustering and the detection of stable clusterings in a graph. However, previous work focuses on using heuristic methods to find local minima of graph clustering objectives, which come with no global approximation guarantees. 

In this paper we establish a theoretical framework for graph clustering in all parameter regimes. Our results come with rigorous guarantees both in terms of the approximation factor, as well as in terms of the number of clusterings needed to approximate an objective in all parameter regimes. We specifically consider the LambdaPrime objective, which is closely related (and even equivalent at optimality) to a number of previously introduced parametric graph clustering objectives. Our contributions include both novel techniques for finding small approximating families of clusterings, as well as fundamental lower bounds on the number of clusterings needed to exactly or approximately solve an objective in all parameter regimes for certain simple graph classes (ring graphs and star graphs). 

All of our techniques for obtaining approximate graph clustering solutions are built on solving and rounding linear programming relaxations of the NP-hard LambdaPrime objective. In future work, we wish to  extend our analysis of the FE algorithm (Section~\ref{sec:fe}) to develop techniques for finding provably small approximating families, without having to evaluate the LP many times and then refine the output. We also wish to better understand upper and lower bounds on the number of clusterings needed to approximate or exactly solve the LambdaPrime objective on other special classes of graphs.

\bibliography{itcsbib}

\appendix
\section{Technical Details for LambdaPrime Variants}
\label{app:lamcc}
In the main text of our manuscript, we focus exclusively on the LambdaPrime objective, a parametric version of Correlation Clustering inspired by the LambdaCC graph clustering framework. In this appendix, we address how to adapt our proof techniques for LambdaPrime, to show similar results for LambdaCC. We also address how to obtain results for node-weighted variants of LambdaPrime and LambdaCC.

\subsection{Correlation Clustering}
Correlation Clustering~\cite{Bansal2004correlation} is a framework for partitioning datasets based on qualitative information about which pairs of objects are similar and which are dissimilar. This is typically cast as an optimization problem over signed graphs. A general instance is given by a signed graph $G = (V, W^+, W^-)$ where for each pair of distinct nodes $(i,j) \in V \times V$ we have similarity and dissimilarity scores $w_{ij}^+ \in W^+$ and $w_{ij}^- \in W^-$ respectively. Both types of weights are nonnegative, but can be interpreted as positive and negative edges between nodes in $G$. The goal of Correlation Clustering is to find a way to partition the graph in a way that minimizes the weight of positive edges between clusters, and negative edges inside clusters. More formally, the objective can be cast as an integer linear program (ILP):
\begin{equation}
\label{eq:gencc}
\begin{array}{lll} \text{minimize}  & \sum_{ij} w_{ij}^+ x_{ij} + w_{ij}^- (1-x_{ij}) &\\ \subjectto  & x_{ij} \leq x_{ik} + x_{jk} & \text{ for all $i,j,k$} \\ & x_{ij} \in \{0,1\} & \text{ for all $i,j$.} \end{array}
\end{equation}
In the above $x_{ij}$ represents the binary \emph{distance} between nodes in a clustering. In other words, if $x_{ij} = 1$, nodes $(i,j)$ are in different clusters, and $x_{ij} = 0$ indicates they are clustered together. Throughout the manuscript, we will repeatedly use the fact that clusterings of a graph are in one-to-one correspondence with feasible solutions to the Correlation Clustering ILP. Note that the objective involves a penalty term $w_{ij}^+ x_{ij}$ that measures the weight of positive edges that cross between clusters. These are \emph{positive mistakes}. The term $w_{ij}^-(1-x_{ij})$ counts the weight of a negative edge that is placed inside a cluster; i.e.\ a \emph{negative mistakes}. Thus, ILP~\eqref{eq:gencc} seek to minimize the weight of mistakes or \emph{disagreements}. A related goal of maximizing agreements for Correlation Clustering also exists, though here we focus on the minimizing disagreements version.
	
\subsection{LambdaCC and LambdaPrime}
Both LambdaPrime and LambdaCC cluster a graph $G = (V,E)$ by constructing a new signed graph and clustering it based on the Correlation Clustering objective. The standard version of LambdaCC replaces edges in $G$ with positive edges of weight $(1-\lambda)$, and replaces non-edges in $G$ with negative edges of weight 1. In other words, LambdaCC defines an instance of Correlation Clustering $\tilde{G} = (V, W^+, W^-)$ with the following weights:
\begin{align}
\label{standardlcc-weights}
\textbf{(LambdaCC weights)} \hspace{.5cm} 
(w_{ij}^+, w_{ij}^-) = \begin{cases}
(1 - \lambda,0) & \text{ if $(i,j) \in E$} \\
(0, \lambda) & \text{ if $(i,j) \notin E$}\,.
\end{cases}
\end{align}
The LambdaPrime framework is very similar, but instead sets weights $(w_{ij}^+, w_{ij}^-)$ as follows:
\begin{align}
\label{standardlamprime-weights}
\textbf{(LambdaPrime weights)} \hspace{.5cm} 
(w_{ij}^+, w_{ij}^-) = \begin{cases}
(1,\lambda) & \text{ if $(i,j) \in E$} \\
(0, \lambda) & \text{ if $(i,j) \notin E$}\,.
\end{cases}
\end{align}
This weighted variant was previously considered in Veldt's PhD thesis~\cite{veldt2019optimization} as an \emph{alternative} LambdaCC formulation. We have referred to it as LambdaPrime to emphasize that it is very closely related to LambdaCC, but not identical in every way. In contrast to LambdaCC, LambdaPrime adds negative edges with weight $\lambda$ between \emph{every} pair of nodes in the graph $G$. From this view, $\lambda$ represents a small amount of repulsion that is introduced between all nodes, which then acts as an implicit regularizer on cluster size and edge density. In the LambdaCC objective, all edges must have a strictly positive or strictly negative relationship. 
Before moving on we highlight several useful facts about the relationship between the two objectives. Most of these facts are easy to see by inserting the two different types of weights into the Correlation Clustering objective~\eqref{eq:gencc}. For more details and proofs, see Chapter 3 of Veldt's PhD thesis~\cite{veldt2019optimization}. 
\begin{theorem}
	\label{thm:2objs}
For any clustering $\mathcal{C}$, let $\textbf{LamCC}(\mathcal{C},\lambda)$ and $\textbf{LamPrime}(\mathcal{C},\lambda)$ represent the LambdaCC and LambdaPrime objective scores for $\mathcal{C}$. 

\begin{enumerate}
	\item For any clustering $\mathcal{C}$ and any $\lambda$,
	\[\textbf{LamPrime}(\mathcal{C},\lambda)=  \textbf{LamCC}(\mathcal{C},\lambda)+ \lambda|E|.\]
	The same relationship holds for LP relaxations of the objectives.
	\item For any $\lambda \in (0,1)$, LambdaPrime and LambdaCC share the same set of optimal solutions. The same is true for solutions to the LP relaxations.
	\item For any value of $\lambda$, a $c$-approximation algorithm for LambdaCC will be a $c$-approximation (or better) for LambdaPrime, though the converse does not necessarily hold.
	\item The optimal LambdaPrime solution scores strictly increase with $\lambda$. This does not hold for LambdaCC.
\end{enumerate}
\end{theorem}
The second point in Theorem~\ref{thm:2objs} implies that all results from Section~\ref{sec:optimal} on optimal solutions for LambdaPrime also hold for optimal solutions to LambdaCC. However, the third point indicates that approximations are different, and furthermore that LambdaCC is harder to approximate than LambdaPrime. Finally, the fourth point indicates that all results proven for LambdaPrime which relied on the objective being monotonic do not immediately hold for LambdaCC. Despite these differences, we can still obtain the same basic results for LambdaCC as we obtained for LambdaPrime. In the following sections, we note changes that must be made in order to adapt our results to the LambdaCC setting.

\subsection{Approximating the LP Relaxation}
In order to obtain a $(1+\varepsilon)$-approximation for the LambdaCC LP relaxation, we can slightly adapt the proof of Lemma~\ref{lem:approx} to show:
\begin{lemma}
	\label{lem:approx2}
	Let $(x_{ij}^t)$ and $(x_{ij}^{t+1})$ be optimal solutions to the LambdaCC LP relaxation for resolution parameters $\lambda_t < \lambda_{t+1}$. Let $\delta = \frac{\lambda_{t+1}}{\lambda_t}\frac{(1-\lambda_t)}{(1-\lambda_{t+1})}$. Then $(x_{ij}^t)$ is a $\delta$-approximate solution for the LambdaPrime LP relaxation when $\lambda_{t+1}$ is used, and $(x_{ij}^{t+1})$ is a $\delta$-approximate solution for the LP relaxation when $\lambda_t$ is used.
\end{lemma}
\begin{proof}
		We will prove that $(x_{ij}^{t+1})$ is a $\delta$-approximation for LambdaCC when $\lambda = \lambda_{t}$. Similar steps will yield the other direction. For $k \in \{t, t+1 \}$, define
		\begin{align*}
		P_k = \sum_{(i,j) \in E} x_{ij}^k \text{ and } N_k = \sum_{i<j} (1-x_{ij}^k).
		\end{align*}
		For arbitrary $\lambda$, the LambdaCC score for $(x_{ij}^k)$ is $(1-\lambda)P_k + \lambda N_k$. 
		Solutions $(x_{ij}^t)$ and $(x_{ij}^{t+1})$ are optimal for $\lambda_t$ and $\lambda_{t+1}$ respectively. Note that
		\begin{align*}
		(1-\lambda_t) P_{t+1} + \lambda_{t} N_{t+1} &\leq \frac{(1-\lambda_{t})}{(1-\lambda_{t+1})} \left( (1-\lambda_{t+1})P_{t+1} + \lambda_{t+1} N_{t+1} \right)\\
		& \leq \frac{(1-\lambda_{t})}{(1-\lambda_{t+1})} \left( (1-\lambda_{t+1})P_{t} + \lambda_{t+1}N_t ) \right)\\
	 & \leq \delta \left( (1-\lambda_t) P_t + \lambda_t N_t \right) .
		\end{align*}
	\end{proof}

For $\varepsilon > 0$, first define a sequence of values between $\frac{1}{n^2}$ to $1$ as follows:
\begin{align}
\label{eq:gammas1}
\gamma_1 &= \frac{1}{n^2}\\
\label{eq:gammas2}
\gamma_{i+1} &= (1+\varepsilon) \gamma_{i} \text{ for $i = 1,2, \hdots q-1$},
\end{align}
where $q = \lceil \log_{(1+\varepsilon)} n^4\rceil + 1$ so that $\gamma_{q+1} \geq n^2$. Then define a sequence of $\lambda$ values by setting 
\begin{equation}
\label{newlambda}
\lambda_{i+1} = \frac{\gamma_{i+1}}{1+\gamma_{i+1}} = \frac{(1+\eps)\frac{\lambda_i}{1 - \lambda_{i}}}{ 1 + (1+\eps)\frac{\lambda_i}{1 - \lambda_{i}} } 
= \frac{(1+\eps)\lambda_i}{ (1 - \lambda_i) + (1+\eps)\lambda_i}
= \frac{(1+\eps)\lambda_i}{1+ \eps \lambda_i}.
\end{equation}
Note that $\lambda_1 = \frac{1}{n^2 + 1} < \frac{4}{n^2}$ and $\lambda_q > n^2/(1+n^2)$, so this covers the entire range of non-trivial $\lambda$ values. Evaluating the LambdaCC LP relaxation at each value of $\lambda_i$ will then provide an approximate solution in every parameter regime, where by Lemma~\ref{lem:approx2}, the worst-case approximation factor is
\begin{equation}
\delta = \frac{\lambda_{i+1}}{\lambda_{i}} \frac{1-\lambda_{i}}{1-\lambda_{i+1}} = \frac{\gamma_{i+1}}{\gamma_i} = (1+\varepsilon).
\end{equation}

Thus, by evaluating the LP relaxation $p = \lceil \log_{(1+\varepsilon)} n^4\rceil + 1$ times, we can get a $(1+\varepsilon)$-approximation for LambdaCC in all parameter regimes.

\subsection{Lower Bound on Ring Graphs for LambdaCC}
Just as we did for LambdaPrime, we can prove that in order to get an approximating family of solutions for the LambdaCC LP relaxation on ring graphs, we will need at least $\Omega(\log n)$ feasible solutions. Here we outline slight differences in the results for LambdaCC, e.g. differences in constant factors when approximating the LP relaxation by other functions. We omit proof details, as they follow by applying the same basic steps we used for LambdaPrime.

Recall that the LambdaCC and LambdaPrime objectives differ by an additive constant $\lambda |E|$, where $|E| = n$ for the ring graph. Our proof for LambdaPrime relied on proving a special way to characterization the LP relaxation, and then bounding it in terms of simpler functions. We will define similar functions for LambdaCC, using an asterisk $(^*)$ to denote functions defined for LambdaCC. Note that throughout, we consider the range $\lambda \in (4/n^2, 1/2)$, since all non-trivial clustering solutions to the ring graph fall in this interval.

For LambdaPrime we proved the LP relaxation score on ring graphs can be written
\begin{equation}
\lp(\lambda) = \min_{t \in \mathbb{N} }  \frac{n}{t} \left( 1 + \lambda {t \choose 2} \right). 
\end{equation}
Since the the objectives differ by an additive constant $\lambda n$ but share the same optimizers, we note that the LP relaxation score for LambdaCC is
\begin{align*}
\lp^*(\lambda) = \min_{t \in \mathbb{N} }  \frac{n}{t} \left( (1-\lambda) + \lambda {t \choose 2} - \lambda(t-1) \right).
\end{align*}
For LambdaPrime we proved the following lower bound on $\lp$:
\begin{equation}
g(\lambda) = \min_{t \in \mathbb{R} } = n \left( \sqrt{2\lambda} - \frac{\lambda }{2}\right) \leq \lp(\lambda).
\end{equation}
With very minor changes, we can adapt Lemma~\ref{lem:g} to prove a lower bound on the LambdaCC relaxation:
\begin{equation}
g^*(\lambda) = n \left( \sqrt{2\lambda} - \frac{3\lambda}{2} \right) \leq \lp^*(\lambda).
\end{equation}
Note that for all $\lambda \in (4/n^2,1/2)$, $g^*$ can be bounded above and below in terms of the square root function as follows:
\begin{equation}
\label{sqrt}
 n\sqrt{2}\sqrt{\lambda} \geq g^*(\lambda) \geq \frac{n \sqrt{2}\sqrt{\lambda}}{4}. 
\end{equation}
In Lemma~\ref{lemma:gbounds} we bounded the LP relaxation value and optimal solution value for LambdaPrime in terms of $g$:
\begin{equation}
g(\lambda) \leq \lp(\lambda) \leq \opt(\lambda) \leq \sqrt{2} g(\lambda)\,.
\end{equation} 
Following the same proof steps, we prove a corresponding bound for the LambdaCC framework.
\begin{equation}
\label{gstarbounds}
g^*(\lambda) \leq \lp^*(\lambda) \leq \opt^*(\lambda) \leq (2\sqrt{2}-1) g^*(\lambda).
\end{equation} 
Setting $q^*(\lambda) = \frac{n \sqrt{2}}{4} \sqrt{\lambda}$, then combining~\eqref{sqrt} with~\eqref{gstarbounds} yields
\begin{equation}
\label{finalbounds}
q^*(\lambda) \leq \lp^*(\lambda) \leq \opt^*(\lambda) \leq 4(2\sqrt{2}-1) q^*(\lambda).
\end{equation}
Applying Lemma~\ref{lemma:hardtransfer} and the results of Magnanti and Stratila~\cite{magnanti2012separable} in Theorem~\ref{thm:magnanti}, we conclude that $\Omega(\log n)$ feasible solutions are needed to approximate the LambdaCC LP relaxation in all parameter regime. The number of feasible solutions differs only by a constant, as compared with the result for LambdaPrime.

\subsection{Applying the FE Algorithm to LambdaCC}
The FE algorithm and its analysis can also be adapted to apply to the LambdaCC framework. First of all, we can slightly change our techniques in Section~\ref{orlp} for finding the optimality range and the $\eps$-approximation range for an LP solution, to account for the slight difference between the LambdaCC and LambdaPrime LP objectives. In LP~\textbf{P1}, we change $\mbc^T \mbx + \lambda {n \choose 2}$ to $\mbc^T \mbx + \lambda ({n \choose 2} -|E|)$, and make minor adjustments to remainder of Section~\ref{orlp} accordingly. Once \textbf{ORLP} has been adapted so that we can compute $\eps$-approximation ranges for LambdaCC, we can define a new FE algorithm for LambdaCC (FE-LamCC). If $\textbf{x}$ is the optimal LambdaCC LP solution at some $\lambda_0$, we can use the new version of \textbf{ORLP} to find $\lambda_+$, the largest value of $\lambda$ for which  $\textbf{x}$ is still an $\eps$-approximation. We can then solve the LambdaCC LP relaxation at a new value $\lambda_1 = (1+\eps)\lambda_+/(1+\eps \lambda_+)$. Here, $\lambda_1$ is defined as in~\eqref{newlambda}, in such a way that by Lemma~\ref{lem:approx2}, the LP solution at $\lambda_1$ is a $(1+\eps)$-approximation in the range $[\lambda_+, \lambda_1]$. The FE-LamCC algorithm continues in this way until it has an $\eps$-cover for the entire range of $\lambda$ values. The remaining results regarding the performance of the FE-LamCC algorithm also follow after similar adjustments. This includes a guarantee that we can refine the FE-LamCC cover in order to find a sub-family of LP solutions that is still an $\eps$-cover for the LambdaCC LP, and contains at most two times the number of LP solutions in a minimum $\eps$-cover.

\subsection{Node-Weighted Variants of LambdaPrime}
We briefly discuss how our approximation results for LambdaPrime can be adapted for node weighted variants. For a graph $G = (V,E)$ that we wish to cluster, consider assigning a weight $\pi(v)$ for every node $v \in V$. When converting the graph into an instance of Correlation Clustering, we can introduce a negative edge between every pair of nodes $(i,j) \in V \times V$ with weight $\lambda \pi(i) \cdot \pi(j)$ for a resolution parameter $\lambda \in (0,1)$.

When $\pi(v) = 1$ for every $v \in V$, this reduces to the standard definition of LambdaPrime we considered throughout the manuscript. Another natural choice is to set $\pi(v) = d_v$, where $d_v$ is the degree of node $v$, i.e., the number of other nodes it is adjacent to. In this case, the weighted LambdaPrime objective shares the same set of optimal solutions with the degree-weighted version of LambdaCC~\cite{Veldt:2018:CCF:3178876.3186110}. The LambdaPrime objective for a clustering $\mathcal{C}$ can be written
\begin{equation}
\label{lpweighted}
\textbf{LamPrime-Weighted}(\mathcal{C}) = \frac{1}{2}\sum_{S \in \mathcal{C}} \cut(S) - \frac{\lambda}{2} \sum_{S \in \mathcal{C}}
\textbf{vol}(S) \textbf{vol}(\bar{S}) + c, 
\end{equation}
where $\textbf{vol}(S) = \sum_{v \in S} d_v$ is the volume of a cluster, and $c = \frac{\lambda}{2} \left( \textbf{vol}(V)^2 -\sum_{i \in V} d_i^2\right )$ is a constant with respect to $\lambda$ (see Chapter 3 of Veldt's work~\cite{veldt2019optimization} for a proof). We note that analogous to Theorem~\ref{thm:optimal}, we can show an upper bound of $|E|$ on the number of clusterings needed to optimally solve the degree-weighted LambdaPrime objective in all parameter regimes. In short, this can be shown by noting that as $\lambda$ increases, the number of positive edge mistakes must increase by at least 1 each time we reach a new breakpoint in the underlying ILP. This implies that the same upper bound of $|E|$ clusterings in an optimal family holds for generalizations of modularity with a resolution parameter~\cite{Arenas2008analysis,ReichardtBornholdt2006}, since these are also equivalent to LambdaPrime and LambdaCC with regard to optimal solutions.

\end{document}